\documentclass[a4paper]{article}
\usepackage{a4wide}
\usepackage[colorlinks=true,linkcolor=blue,citecolor=red]{hyperref}
\usepackage{xcolor}
\usepackage{mathtools} 
\usepackage{amsthm, amssymb} 
\usepackage{bm}
\usepackage{appendix}
\usepackage[pagewise]{lineno}

\theoremstyle{Theorem}\newtheorem{proposition}{Proposition}[section]
\theoremstyle{Remark}\newtheorem{remark}{Remark}[section]
\theoremstyle{Theorem}\newtheorem{theorem}{Theorem}[section]

\DeclarePairedDelimiter\ave{\langle}{\rangle}

\newcommand{\R}{\mathbb{R}}

\newcommand{\hv}{\hat{v}}
\newcommand{\tv}{\tilde{v}}
\newcommand{\hf}{\hat{f}}
\newcommand{\tf}{\tilde{f}}
\newcommand{\barf}{\rho}
\newcommand{\hphi}{\hat{\varphi}}
\newcommand{\tphi}{\tilde{\varphi}}

\newcommand{\hV}{\hat{V}}
\newcommand{\tV}{\tilde{V}}
\newcommand{\hmu}{\hat{\mu}}
\newcommand{\tmu}{\tilde{\mu}}
\newcommand{\intS}{\int_{\mathbb{S}^{d-1}}}
\newcommand{\intR}{\int_{\R_+}}
\newcommand{\intV}{\int_{\mathcal{V}}}
\newcommand \commentout[1] {}

\newenvironment{sistem}%
{\left\lbrace\begin{array}{@{}l@{}}}%
{\end{array}\right.}

\allowdisplaybreaks

\begin{document}
\title{A novel linear transport model with distinct scattering mechanisms for direction and speed}
\author{Martina Conte\thanks{\texttt{Corresponding author: martina.conte@polito.it}} \and Nadia Loy\thanks{\texttt{nadia.loy@polito.it}}}
\date{\small Department of Mathematical Sciences ``G. L. Lagrange'' \\ Politecnico di Torino, Italy}

\maketitle

\begin{abstract}
We introduce a novel linear transport equation that models the evolution of a one-particle distribution subject to free transport and two distinct scattering mechanisms: one affecting the particle's speed and the other its direction. These scattering processes occur at different time scales and with different intensities, leading to a kinetic equation where the total scattering operator is the sum of two separate operators. Each of them depends not only on the kernel characterizing the corresponding scattering mechanism, but also explicitly on the marginal distribution of either the speed or the direction. Therefore, unlike classical settings, the gain terms in our operators are not tied to a fixed equilibrium distribution but evolve in time through the marginals. As a result, typical analytical tools from kinetic theory, such as equilibrium characterization, entropy methods, spectral analysis in Hilbert spaces, and Fredholm theory, are not applicable in a standard fashion. In this work, we rigorously analyze the properties of this new class of scattering operators, including the structure of their non-standard pseudo-inverses and their asymptotic behavior. We also derive macroscopic (hydrodynamic) limits under different regimes of scattering frequencies, revealing new effective equations and highlighting the interplay between speed and directional relaxation.


\medskip

\noindent{\bf Keywords:} Boltzmann-type equations, Markov jump processes, transition probabilities, transport equation, scattering mechanisms  
\medskip

\noindent{\bf Mathematics Subject Classification:} 35Q20, 35Q70, 35Q84
\end{abstract}




\section{Introduction}
The linear  transport equation, or linear Boltzmann equation, describes the evolution of a one-particle distribution which moves by free drift and undergoes scattering due to an external medium. 
Given the one particle distribution $f=f(t,x,v)$ defined for each time $t\ge 0$ on the phase space $(x,v)$, being $x\in \R^d$ the position and $v \in \mathcal{V}\subset \R^d$ the microscopic velocity, a typical form is
\begin{equation}\label{kin.eq:M}
    \partial_t f +v\cdot \nabla_x f =\mu (\rho M(v)-f)
\end{equation}
being $\rho$ the mass density. While the operator $v\cdot \nabla_x f$ implements the free drift, the stochastic process underlying the scattering operator in the right-hand side in~\eqref{kin.eq:M}, often called \textit{velocity-jump process}, is a piecewise Markov process ruled by a Poissonian renewal process of intensity (frequency) $\mu$~\cite{gardiner,Risken1996}. The new velocity $v$ is chosen according to the \textit{transition probability} (or \textit{kernel}) $M(v)$ which describes the fixed background that imposes the new reorientation velocities to the moving particle. Being a transition probability, $M:\R^d \to \R_+$ satisfies a normalization condition $\intV M(v) \, {\rm d}v =1$. Here we are assuming the simplification in which the scattering process is a memoryless Markovian process, i.e., $M$ does not depend on the previous velocity. Equation~\eqref{kin.eq:M} is sometimes referred to as a BGK-type transport equation~\cite{Cercignani,BC}.  

The linear Boltzmann equation is originally a simplification of the full Boltzmann equation, where interactions between particles are linearized, making it applicable to situations with low particle densities and weak interactions~\cite{Cercignani}. This equation has a large variety of long standing applications ranging from  gas dynamics and plasma physics, to radiation transfer and dust
particles~\cite{Cercignani,lods2004JSP,dautray1990}, but also more recent applications such as traffic flow~\cite{puppo2017KRM,visconti2017MMS}, bacteria~\cite{calvez2015KRM,bouin2015ARMA} and cell migration~\cite{stroock1974ZWVG,alt1980JMB,hillen2000SIAP,hillen2006JMB}. 

Mathematically, this equation has been widely studied, both from the point of view of existence and uniqueness of the solution~\cite{pettersson1983existence}, as well as decay to equilibrium~\cite{pettersson1992convergence,pettersson1993weak,Bisi2014EntropyDE,canizo2020KRM,desvillettes2006}. Macroscopic (or hydrodynamic) limits in diffusive, hyperbolic, and fractional regimes have also been derived (see e.g.~\cite{hillen2000SIAP,hillen2006JMB,degond2000IUMJ,mellet2008ARMA}). 

One of the main mathematical issues is the interplay between the free transport operator $v\cdot \nabla_x$ and the scattering process embodied by the following \textit{scattering operator} defined on the space of $L^2$ functions on $\mathcal{V}$:
\begin{equation}\label{operator.proto}
L: L^2(\mathcal{V}) \longmapsto L^2(\mathcal{V}), \qquad Lf=\rho M(v) -f,
\end{equation}
which conserves the mass as $M$ is normalized to one. The transition probability $M$ defines the \textit{equilibrium} of~\eqref{kin.eq:M} making the operator vanish, as its kernel is $\ker(L) = {\rm span} \lbrace M\rbrace$. The operator~\eqref{operator.proto} is typically studied on the following $L^2$ Hilbert space
\begin{equation*}\label{Hilbert.proto}
L_M^2(\mathcal{V}):= \left\{ \varphi \in L^2(\mathcal{V}) : || \varphi||^2_{L^2_M}:= \intV \varphi^2(v) \dfrac{1}{M(v)} \, {\rm d} v <  \infty \right\},
\end{equation*}
where the Fredholm alternative allows to state that the $L_M^2(\mathcal{V})$ space can be orthogonally partitioned between the equilibria carrying all the mass and the functions with zero mass, i.e.
\begin{equation}\label{L2M}
 L_M^2(\mathcal{V}) = {\rm span} (M) \oplus   {\rm span} (M)^{\bot}
\end{equation}
being
\begin{equation}\label{spanM}
    {\rm span} (M)^{\bot} = \left\{ \varphi \in L^2(\mathcal{V}) : \intV \varphi \, {\rm d} v = 0 \right\}
\end{equation}
the space of functions with vanishing mass.
As a consequence,  for $f \in {\rm span} (M)^{\bot}$,  $Lf$ has no gain term as $\rho=0$; therefore~\eqref{operator.proto} can be easily inverted on $ {\rm span} (M)^{\bot}$.

As a matter of fact, equation~\eqref{kin.eq:M} describes a scattering process in which the whole microscopic velocity $v$ changes, i.e., both direction and speed change at times ruled by the same Poisson process of intensity $\mu$ and according to a transition probability which essentially regulates the reorientation mechanism. However, prototype migration modes concerning cell motility have shown the importance of having two different backgrounds causing the scattering of the direction  $\hv \in \mathbb{S}^{d-1}$, and of the speed $\tv \in \R_+$.  In fact, cells and bacteria move by run-and-tumble, i.e., they run over straight lines and then reorientate changing their direction of motion $\hv$, that may be chosen according to an external bias.  This has been effectively modeled using velocity-jump processes by means of kinetic equations such as~\eqref{kin.eq:M}, as presented in ~\cite{stroock1974ZWVG,hillen2006JMB,calvez2015KRM,Filbet_Perthame}. In this context, $M$ describes the reorientation probability, often affected by the presence of a directional, or \textit{tactic}, cue, such as a chemoattractant or the distribution of fibers in a directed environment. However, cells may undergo not only \textit{taxis}, i.e. `directed orientation reactions', but also \textit{kineses}, which is defined as an `undirected locomotory reaction  in which the speed of movement (...) depends on the intensity of
stimulation'~\cite{Doucet1990}. Indeed, the speed may change because of external factors, which may be cues of chemical but also mechanical origin, such as a viscous environment or the presence of other cells. Moreover, many tactic and kinetic cues may be present at the same time in the cellular environment. A prototype example is represented by cell migration on the \textit{extracellular-matrix} (ECM), which is of the utmost importance as it regulates most of the pathological and physiological processes, from wound healing to cancer spread~\cite{Taufalele2019PLOSONE,provenzano2008collagen}. The ECM is the non-cellular component of the human organism which is composed mainly by a fibrous collagen. 
Cells moving on the ECM reorientate because of the presence of fibers on which they crawl, and their speed may be modulated by the matrix density, stiffness, viscosity and porosity~\cite{charras2014nat}.
In this case, an appropriate modeling choice would be  $\mathcal{V}=\R_+\times\mathbb{S}^{d-1}$ so that $M$ depends on the couple $(\tv,\hv)$ instead of the velocity vector $v=\tv\hv$. Specifically, many individual microscopic mechanisms in cell migration, among which migration on ECM, have been reproduced by considering a scattering kernel of the form
\begin{equation}\label{def:M}
    M(\tv,\hv)=\psi(\tv|\hv) q(\hv)\,.
\end{equation}
In~\eqref{def:M}, $q(\hv)$ describes a directed fixed background related to the change of the orientation $\hv$ (e.g., in the case of ECM, the fibers distribution, as introduced in~\cite{hillen2006JMB}, which can be estimated from experimental data). Instead, $\psi(\tv|\hv)$  describes the distribution of the new speeds $\tv$ on a fixed migration direction~\cite{loy2019JMB,conte2023SIAP}, which may be modulated, for example, by the ECM density and that can be also inferred from experiments.   

Furthermore, experimental observations show that not only there may be superposing tactic and kinetic cues in the same environment, but also taxis and kineses in response to different factors may happen at \textit{distinct times} and with \textit{different intensities}. In fact, the physico-mechanical cues of the ECM may all feature a different variability with respect to the fibrous structure which regulates the direction of motion~\cite{Taufalele2019PLOSONE,Condor,Lang,Yamada2019IJEP}. Then, not only different cues affect the dynamics of the direction and of the speed, but also different time scales are involved in the two processes. For instance, for breast cancer the ECM is considered a prognostic factor as it plays a promoting role in local invasion, both through tactic~\cite{provenzano2006collagen} and kinetic~\cite{provenzano2008collagen} cues. In the context of brain tumors, for which the key spread process is migration along the fibrous white matter measurable with MRI data~\cite{conte2020glioma}, novel therapeutic strategies can be seen as acting on the motility of cells by reducing their speed along favorable directions and promoting a more isotropic spread \cite{brooks2021white}.

Motivated by this, we propose a new kinetic model implementing distinct stochastic processes for the speed and for the direction. These two superposing microscopic dynamics give rise to a transport equation featuring an operator given by is the sum of two separate scattering operators involving the marginal distributions of $f$ along with the scattering kernels for the two mechanisms. Specifically, our novel transport equation will not be in the form~\eqref{kin.eq:M}, in which the gain term is essentially defined by the equilibrium. Instead, it will be a time-dependent (through the marginals) gain term which is not directly defined by a fixed equilibrium distribution. This technical issue will pose many challenges which extend beyond the determination of the explicit equilibrium of the operators, their sum, and the \textit{marginal operators} appearing in the equations for the marginals. In fact, the marginal operators will not map the same functional spaces, so that even though the marginal equilibria allow to define scalar products and Hilbert spaces and the subsequent Fredholm alternative, not all gain terms vanish on the orthogonal space to the one generated by the equilibrium. This will require to define non-standard pseudo-inverse operators. Furthermore, our approach allows to investigate macroscopic limits in different regimes of time scales for the two processes. Some of these issues are marginally tackled in~\cite{lorenzi2025}, where the migrating particle has a microscopic velocity and a phenotype, and the transport equation for the one particle distribution features two distinct operators describing the dynamics of the velocity vector and of the phenotype. In the present work, the truly non-standard difficulty is due to the evolution of the speed and of the direction caused by two different processes.

In section \ref{sec:micro_des} we shall propose a formal derivation of the novel kinetic equation~\eqref{eq:kin_weak} from discrete-time stochastic processes. We study the strong equation~\eqref{eq:kin_strong} with semigroup theory and use the results in order to give an insight on the continuous stochastic trajectories of the particles. Then, in section \ref{sec:operators}, we shall focus on the properties of the scattering operator, on its kernel, on the Fredholm-alternative and on their inversion. Some preliminary results on the entropy decay in the homogeneous case will be presented. In section \ref{sec:macro_lim}, we shall investigate macroscopic limits of the kinetic equation, by considering different scalings defined by different orders of frequencies of the two processes. Finally, the results will be illustrated by means of numerical tests in section \ref{sec:num_tests}.  


\section{Microscopic stochastic description}\label{sec:micro_des}
In this section we present a formal derivation of a kinetic equation implementing a microscopic dynamics in which the speed and direction change at times that are ruled by two independent stochastic processes. The microscopic dynamics is described in terms of discrete-time stochastic processes, as it is customary in the literature of multi-agent systems~\cite{pareschi2013BOOK}. The analysis of the kinetic equation by means of semigroup theory allows to give an insight on the continuous-time stochastic individual trajectories.
\subsection{Formal derivation from discrete-time stochastic processes}
 Let us consider a particle identified by its position $x\in \mathbb{R}^d$, speed $\tv \in \R_+$, and direction of motion $\hv \in \mathbb{S}^{d-1}$, being $d$ the space dimension.  The notation $v:= \tv \hv \in \mathbb{R}^d$
defines $v$ to be the velocity vector which clearly has magnitude $\tv$ and direction (or orientation) $\hv$.

We then introduce the corresponding random variables $X_t \in \R^d $, $ \tV_t \in \R_+$, $\hV_t \in \mathbb{S}^{d-1} $ describing the position, speed and direction, respectively. Given a time step $\Delta t >0$ we consider the discrete-time domain $D_{\Delta t}:=\lbrace k\Delta t, k\in \mathbb{N}\rbrace$, and we express the stochastic processes 
\begin{align}
&X_{t+\Delta t}=X_t+\Delta t \tV_t \hV_t,\label{def:micro.X}\\[0.2cm]
&\tV_{t+\Delta t}=(1-\tilde{\Theta}_{t+\Delta t}) \tV_t+\tilde{\Theta}_{t+\Delta t}\tV_{t+\Delta t}',\label{def:micro.V}\\[0.2cm]
&\hV_{t+\Delta t}=(1-\hat{\Theta}_{t+\Delta t})\hV_t+\hat{\Theta} \hV_{t+\Delta t}',\label{def:micro.hV}
\end{align}
that are defined for $t\in D_{\Delta t}$ given the initial condition $(X_0,\tV_0,\hV_0)=(x_0,\tv_0,\hv_0)$. The stochastic process~\eqref{def:micro.X} describes the time evolution of the position $X_t$ at time $t$ into a new one $X_{t+\Delta t}$ at time $t+\Delta t$ that is determined by the free drift performed by the particle during the time interval $\Delta t$ along the constant direction $\hV_t$ and with constant speed $ \tV_t$, being $\hV_t, \tV_t$ the direction and speed at time $t$, respectively. The time evolution of the speed and of the direction are both due to jump-processes that we assume to happen at stochastically \textit{independent times}. Specifically, the direction changes in consequence of a \textit{reorientation} process (also called \textit{direction-jump}), which depends on an external cue, while the speed changes according to a \textit{speed-jump} process influenced by the external environment, which may modulate the speed of motion. Being the time interval $\Delta t$ fixed, both the speed $\tV_{t+\Delta t}$ and the direction $\hV_{t+\Delta t}$ at time $t+\Delta t$ may change or not into a new one with respect to the departing speed $ \tV_t$ and direction $\hV_t$. In~\eqref{def:micro.V}-\eqref{def:micro.hV} this is expressed by $\tilde{\Theta}_{t+\Delta t}$ and $\hat{\Theta}_{t+\Delta t}$ which are Bernoulli random variables that are sampled at each time iteration $t+\Delta t, t\in D_{\Delta t}$. For each sample time $t \in D_{\Delta t}$ we assume that
\begin{equation}\label{def:Berny}
\tilde{\Theta}_t \sim {\rm Bernoulli} ( \tmu \Delta t), \quad \hat{\Theta}_t \sim {\rm Bernoulli} (\hmu \Delta t), \qquad \forall t \in D_{\Delta t}
\end{equation} 
where $\tmu, \hmu\in \R_{+}$ are the positive frequencies of the two Bernoulli trial processes. We here remark that for the good definition of both Bernoulli random variables, we shall need a restriction on the time step, i.e., $\Delta t \le \min \lbrace \frac{1}{\tmu}, \frac{1}{\hmu}\rbrace$.
$\tilde{\Theta}_t, \hat{\Theta}_t$ are assumed to be independent, in order to embody the independence of the times of the two jump processes, i.e., 
\begin{equation}\label{ass:indep1}
\tilde{\Theta}_t, \quad \hat{\Theta}_t \qquad {\rm are} \, {\rm independent} \, \forall t \in D_{\Delta t}.
\end{equation}
Moreover, we assume that
\begin{equation}\label{ass:indep2}
\tV_t, \quad \tilde{\Theta}_t \qquad {\rm are} \,\, {\rm independent} \,\, \forall t \in D_{\Delta t},
\end{equation}
and that
\begin{equation}\label{ass:indep3}
\hV_t, \quad \hat{\Theta}_t \qquad {\rm are} \,\, {\rm independent} \,\, \forall t \in D_{\Delta t}.
\end{equation}

In~\eqref{def:micro.V}, $\tilde{\Theta}_{t+\Delta t}$ discriminates whether a speed-jump happens ($\tilde{\Theta}_{t+\Delta t}=1$) or not ($\tilde{\Theta}_{t+\Delta t}=0$). If it does, $\tV_{t}$ changes into $\tV_{t+\Delta t}=\tV_{t+\Delta t}'$, where $\tV_{t+\Delta t}'$ is a random variable (sampled at time $t+\Delta t$). We assume $\lbrace\tV_{t}'\rbrace_{t\in D_{\Delta t}}$ to be a homogeneous Markov process with no memory distributed according to the transition probability $\psi: \R_+ \times \mathbb{S}^{d-1}\to \R_+$, i.e.
\begin{equation}\label{def:psi}
 \tV_t' \sim \psi( \tV_t'=\tv'|\hV_t=\hv), \qquad \forall t \in D_{\Delta t}.
\end{equation}
We highlight that we are considering that the transition probability $\psi$ of $ \tV_t'$ depends on the current direction of motion $\hV_t=\hv$. 
Being $\psi$ a transition probability, we are assuming that
\begin{equation}\label{norm.psi}
\intR \psi(\tv|\hv) \, {\rm d}\tv =1, \qquad \forall \hv \in \mathbb{S}^{d-1},
\end{equation}
in such a way that, given $A \subset\R_+$, we have that ${\rm Prob}( \tV_t'\in A)= \int_A \psi(\tv|\hv) \, {\rm d}\tv.$
In particular, we are assuming that
\begin{equation}\label{psi.L1}
    \psi(\cdot|\hv) \in L^1(\R_+) \qquad \forall \hv \in \mathbb{S}^{d-1}.
\end{equation}
 
Concerning the direction, its time variation is expressed analogously by the stochastic process~\eqref{def:micro.hV}. In~\eqref{def:micro.hV}, $\hat{\Theta}_{t+\Delta t}$ discriminates whether a jump happens ($\hat{\Theta}_{t+\Delta t}=1$) or not ($\hat{\Theta}_{t+\Delta t}=0$). If it does, $\hV_{t}$ changes into $\hV_{t+\Delta t}=\hV_{t+\Delta t}'$, where $\hV_{t+\Delta t}'$ is a random variable (sampled at time $t+\Delta t$). We assume $\lbrace\hV_{t}'\rbrace_{t\in D_{\Delta t}}$ to be a homogeneous Markov process with no memory distributed according to the transition probability $q: \mathbb{S}^{d-1} \to \R_+$, i.e.,
\begin{equation}\label{def:q}
\hV_t' \sim q(\hV_t'=\hv'), \qquad \forall t \in D_{\Delta t}.
\end{equation}
Being $q$ a transition probability, we have that
\begin{equation}\label{norm.q}
\intS q(\hv) \, {\rm d} \hv =1,
\end{equation}
in such a way that, given $A \subset \mathbb{S}^{d-1}$ we have that ${\rm Prob}(\hV_t'\in A)= \int_A q(\hv) \, {\rm d}\hv.$
Note that we are assuming that
\begin{equation}\label{q.L1}
    q(\cdot) \in L^1(\mathbb{S}^{d-1}). 
\end{equation}

In order to formally derive the kinetic equation, we first introduce the kinetic one-particle density function on the phase space
\begin{equation}\label{def:f}
\begin{aligned}[b]
f: \R_+\times \R^d \times \R_+ \times \mathbb{S}^{d-1} &\longmapsto \R_+ \\[0.3cm]
(t,x, \tv, \hv) &\longmapsto f(t,x, \tv, \hv)
\end{aligned}
\end{equation}
that is, at each $(t,x)\in\R_+\times \R^d$ fixed, the joint probability density function of the couple of random variables $\tV_t, \hV_t$.
Le us now consider a test function, which is an observable quantity depending on the microscopic velocity couple, i.e., $\varphi=\varphi(\tv,\hv)$ defined on $\mathcal{V}$, being $\mathcal{V}= \R_+ \times \mathbb{S}^{d-1}$. We assume that $\varphi \in \mathcal{C}_c(\mathcal{V})$ with compact support. The evolution equation for the expected value of $\varphi$ along the stochastic trajectories defined by~\eqref{def:micro.X}-\eqref{def:psi} and~\eqref{def:q} in the limit~$\Delta t \to 0^+$ 
is
\begin{equation}\label{eq:kin_weak}
\begin{split}
&\dfrac{\partial}{\partial t}\intV\varphi(\tv,\hv) f(t,x,\tv,\hv) {\rm d}\hv {\rm d}\tv+\nabla_x \cdot \intV\varphi(\tv,\hv) v f(t,x,\tv,\hv) \, {\rm d}\hv \, {\rm d}\tv=\\[0.2cm]
&\phantom{\dfrac{\partial}{\partial t}\intV\varphi(\tv,\hv) f(t,x,\tv}\hmu\intV\left[\intS q(\hv')\varphi(\tv,\hv')\, {\rm d}\hv'-\varphi(\tv,\hv)\right]f(t,x,\tv,\hv) \,{\rm d}\hv \, {\rm d}\tv\\[0.2cm]
&\phantom{\dfrac{\partial}{\partial t}\intV\varphi(\tv,\hv) f(t,,}+ \tmu\intV\left[\intR\psi(\tv'|\hv)\varphi(\tv',\hv)\, {\rm d}\tv'-\varphi(\tv,\hv)\right]f(t,x,\tv,\hv) \,{\rm d}\hv \, {\rm d}\tv.
\end{split}
\end{equation}
The full formal procedure for deriving~\eqref{eq:kin_weak} is standard and it is reported in the Appendix~\ref{appendix1}.
Equation~\eqref{eq:kin_weak} is the weak kinetic equation for the probability density function $f$ for particles at $(t,x)$ having velocity $(\tv,\hv)$ evolving according to the microscopic stochastic dynamics~\eqref{def:micro.X}-\eqref{def:psi} and~\eqref{def:q} in the continuous-time limit~$\Delta t \to 0^+$. The first term on the right-hand side, i.e., the second line in~\eqref{eq:kin_weak}, describes the change in direction happening with frequency $\hmu$ and ruled by~\eqref{def:q} being $\tv$ fixed, while the second term, i.e., the third line in~\eqref{eq:kin_weak} describes the change in speed happening with frequency $\tmu$ and ruled by~\eqref{def:psi} being the direction $\hv$ fixed.

\subsection{The strong kinetic equation}
Let us now introduce the \textit{number density} $\rho$ of the one particle distribution $f$ as its zero-th order statistical moment for each $(t,x)$ fixed, i.e.,
\begin{equation}\label{def:bf}
\rho \equiv \rho(t,x):=\intV f(t,x,\tv,\hv) \, {\rm d} \tv \, {\rm d}\hv.
\end{equation}
As $f$ is a joint probability function for the couple of random variables $\tV, \hV$ at each $(t,x)$ fixed, we also introduce the normalized marginals. Let us then define the marginal of $\hV$, i.e., the function ${\rho \hf: \R_+\times \R^d \times \mathbb{S}^{d-1}\to \R_+}$ as
\begin{equation}\label{def:hf}
\rho \hf (\hv) \equiv \rho \hf(t,x,\hv):=  \intR f(t,x,\tv,\hv) \, {\rm d} \tv, 
\end{equation}
and the marginal of $\tV$, i.e., $\rho \tf: \R_+\times \R^d \times \R_+ \to \R_+$ as
\begin{equation}\label{def:tf}
 \rho  \tf(\tv) \equiv \rho \tf(t,x,\tv):= \intS f(t,x,\tv,\hv) \, {\rm d} \hv.
\end{equation}
The marginals $\hf, \tf$ are normalized probability density functions in the sense that, by definitions~\eqref{def:hf}-\eqref{def:tf} the following conditions hold:
\begin{equation}\label{norm.marg}
\intS \hf(\hv) \, {\rm d} \hv =1, \qquad \intR \tf(\tv) \, {\rm d} \tv =1.
\end{equation}
We can now write the strong form of the evolution equation~\eqref{eq:kin_weak} for $f$ defined in~\eqref{def:f}, that is
\begin{equation}\label{eq:kin_strong}
\partial_t f +v \cdot \nabla_x f = \hmu \left[\rho\tf(\tv) q(\hv) -f\right] +  \tmu\left[\rho\hf(\hv)\psi(\tv|\hv)-f\right], 
\end{equation}
which we couple with an initial condition
\begin{equation}\label{ci}
    f(0,x,\tv,\hv)\equiv f_0(x,\tv,\hv), \qquad (x,\tv,\hv) \in \R^d\times \R_+ \times \mathbb{S}^{d-1}.
\end{equation}
Equation~\eqref{eq:kin_strong} is a linear transport equation for the probability density $f$, in which in the left-hand side the operator $v\cdot \nabla_x$ describes the free particle drift performed with velocity $v=\tv \hv$. The right-hand side features two different scattering terms. The first one implements the direction-jump dynamics being the speed fixed. In fact, the gain term $\hmu \rho\tf(\tv) q(\hv)$ describes the fraction of particles $\rho \tf(\tv)$ with $\tv$ fixed that with frequency $\hmu$ acquire a new direction $\hv$ according to $q$. The second one is related to the speed-jump process being $\hv$ fixed, and its gain term $\tmu \rho\hf(\hv)\psi(\tv|\hv)$ describes the fraction of particles $\rho\hf(\hv)$ that, having direction $\hv$ fixed, change their speed with frequency $\tmu$ according to the density $\psi(\tv|\hv)$. 
It can be easily verified from~\eqref{eq:kin_weak} setting $\varphi=1$ and integrating over $\R^d$ in ${\rm d}x$ that the total number of particles is conserved in time, i.e.,
\begin{equation}\label{mass.cons}
    \int_{\R^d\times\mathcal{V}} f(t,x,\tv,\hv) \, {\rm d} \tv \, {\rm d} \hv \, {\rm d}x = \int_{\R^d\times\mathcal{V}} f_0(x,\tv,\hv) \, {\rm d} \tv \, {\rm d} \hv \, {\rm d}x.
\end{equation}

For formal purposes, let us now rewrite equation~\eqref{eq:kin_strong} as
\begin{equation}\label{eq:kin_strong.2}
\partial_t f +v \cdot \nabla_x f = \hmu \hat{\mathcal{L}}f +  \tmu\tilde{\mathcal{L}}f,   
\end{equation}
where we have introduced the \textit{direction-jump (or reorientation) operator}
\begin{equation*}\label{def:hL}
\begin{aligned}[b]
&\hat{\mathcal{L}}: L^2 (\mathcal{V}) \longmapsto L^2 (\mathcal{V})\\[0.3cm]
&\hat{\mathcal{L}}\varphi(\tv,\hv):=q(\hv)\intS \varphi(\tv,\hv') \, {\rm d}\hv' -\varphi(\tv,\hv) = \rho_\varphi\tphi(\tv) q(\hv) -\varphi(\tv,\hv),
\end{aligned}
\end{equation*}
and the \textit{speed-jump operator}
\begin{equation*}\label{def:tL}
\begin{aligned}[b]
&\tilde{\mathcal{L}}: L^2 (\mathcal{V}) \longmapsto L^2 (\mathcal{V})\\[0.3cm]
&\tilde{\mathcal{L}}\varphi(\tv,\hv):=\psi(\tv|\hv) \intR \varphi(\tv',\hv) \, {\rm d} \tv' - \varphi(\tv,\hv) = \rho_\varphi\hphi(\hv)\psi(\tv|\hv)-\varphi(\tv,\hv). 
  \end{aligned}
\end{equation*}
The quantities $\rho_\varphi, \tphi, \hphi$ are defined as in~\eqref{def:bf}-\eqref{def:hf}-\eqref{def:tf}, definitions that we report here for completeness
\begin{equation}\label{def:bthphi}
    \rho_\varphi:= \intV \varphi \, {\rm d} \tv \, {\rm d} \hv, \quad  \tphi:= \dfrac{1}{\rho_\varphi}\intS \varphi \,  {\rm d} \hv, \quad \hphi:=  \dfrac{1}{\rho_\varphi}\intR \varphi \,  {\rm d} \tv.
\end{equation}
We remark that in~\eqref{def:bf} we are denoting $\rho_f \equiv \rho$. 

Let us, eventually, define $\mathcal{L}:= \hmu \hat{\mathcal{L}} +  \tmu\tilde{\mathcal{L}}$, that is the following operator
\begin{equation}\label{def:L}
\begin{aligned}[b]
&\mathcal{L}: L^2 (\mathcal{V}) \longmapsto L^2 (\mathcal{V})\\[0.3cm]
&\mathcal{L}\varphi(\tv,\hv):=\big(\hmu\rho_\varphi \tphi(\tv) q(\hv)+\tmu\rho_\varphi \hphi(\hv)\psi(\tv|\hv)\big)-(\hmu+\tmu)\varphi. 
  \end{aligned}
\end{equation}
Thanks to~\eqref{def:L}, ~\eqref{eq:kin_strong.2} can be formally written as a linear transport kinetic equation with scattering operator $\mathcal{L}$, i.e.,
\begin{equation}\label{eq:boltzmann}
\partial_t f +v \cdot \nabla_x f=\mathcal{L}f.   
\end{equation}
However, despite its compact form, the linear equation~\eqref{eq:boltzmann} has only a formal resemblance with the linear Boltzmann equation~\eqref{kin.eq:M}. The main difference between~\eqref{eq:boltzmann}-\eqref{def:L} (i.e.,~\eqref{eq:kin_strong}) and~\eqref{kin.eq:M} lies in the structure of the gain term of the operator $\mathcal{L}$, which is not in the form $\rho M$, where $M$ is the transition probability describing the simultaneous change of the velocity couple. Conversely, the gain term in $\mathcal{L}$ depends on $f$ (and then also on time) through $\hf$ and $\tf$ and not only through $\rho$. 

It is worthwhile to stress again that the marginals $\hf$ and $\tf$ appear in the operators $\tilde{\mathcal{L}}$ and $\hat{\mathcal{L}}$, respectively, as these two operators describe the two different stochastic processes for the speed (that changes being the direction fixed) and for the direction (that changes being the speed fixed).  Let us then consider the equations for the marginals $\tf$ and $\hf$. The evolution equation for $\tf$ is obtained by integrating~\eqref{eq:kin_strong} over $\mathbb{S}^{d-1}$ that gives 
\begin{equation}\label{eq:tf}
    \partial_t (\barf \tf) +\nabla_x \cdot \left( \tv \intS \hv f \, {\rm d} \hv\right) = \tmu \barf\left[ \intS \hf(\hv) \psi(\tv|\hv) \, {\rm d} \hv-\tf \right],
\end{equation}
or, alternatively, by setting $\varphi(\tv,\hv)=\varphi(\tv)\cdot 1$ in~\eqref{eq:kin_weak} and writing the strong form with respect to $\tv$.
Analogously, the evolution equation for $\hf$ is obtained by integrating~\eqref{eq:kin_strong} over $\R_+$ that gives
\begin{equation}\label{eq:hf}
 \partial_t (\barf \hf) +\nabla_x \cdot \left( \hv \intR \tv f \, {\rm d} \tv\right) = \hmu \barf\left[ q(\hv)-\hf \right] ,   
\end{equation}
or, alternatively, by setting $\varphi(\tv,\hv)=\varphi(\hv)\cdot 1$ in~\eqref{eq:kin_weak} and writing the strong form with respect to $\hv$.
The equations for the marginals~\eqref{eq:tf}-\eqref{eq:hf} lead to define the \textit{marginal operators} 
\begin{equation}\label{def:Lpsi}
\begin{aligned}[b]
    &\tilde{\mathcal{L}}_\psi: L^2(\mathcal{V}) \longmapsto L^2(\R_+) \\
    &\tilde{\mathcal{L}}_\psi\varphi(\tv) = \tmu\intS\left(  \psi(\tv|\hv)\intR \varphi(\tv',\hv) \, {\rm d} \tv' -\varphi(\tv,\hv)\right) \, {\rm d} \hv ,
    \end{aligned}
\end{equation}
and
\begin{equation}\label{def:Lq}
    \begin{aligned}[b]
    &\hat{\mathcal{L}}_q: L^2(\mathcal{V}) \longmapsto L^2(\mathbb{S}^{d-1}) \\
    &\hat{\mathcal{L}}_q\varphi(\hv) = \hmu\intR\left(  q(\hv)\intS \varphi(\tv,\hv') \, {\rm d} \hv' -\varphi(\tv,\hv)\right) \, {\rm d} \tv, 
    \end{aligned}
\end{equation}
so that~\eqref{eq:tf}-\eqref{eq:hf} may be rewritten as
\begin{equation}\label{eq:tf.bis}
    \partial_t (\barf \tf) +\nabla_x \cdot \left( \tv \intS \hv f \, {\rm d} \hv\right) = \tilde{\mathcal{L}}_\psi f,
\end{equation}
and 
\begin{equation}\label{eq:hf.bis}
 \partial_t (\barf \hf) +\nabla_x \cdot \left( \hv \intR \tv f \, {\rm d} \tv\right) = \hat{\mathcal{L}}_qf.   
\end{equation}
We remark that the equations for the marginals~\eqref{eq:tf.bis}-\eqref{eq:hf.bis} are not closed in $\tf,\hf$ respectively. One reason lies in the presence of the transport term, that, as usual, involves higher order statistical moments of $f$. In fact,~\eqref{eq:tf.bis} features the average direction of $f$ being the transport speed $\tv$ fixed, while~\eqref{eq:hf.bis} features the average speed of $f$ being the transport direction $\hv$ fixed. However, even in the spatially homogeneous case, i.e., no transport term involved, the closure of~\eqref{eq:tf.bis}-\eqref{eq:hf.bis} would not be granted. In fact, the microscopic dynamics for the speed is ruled by the transition probability $\psi$ that depends on $\hv$. This is why the operator $\tilde{\mathcal{L}}_{\psi}$ is defined on $f$, because, as clearly written in~\eqref{eq:tf}, it features both marginals, being $\tilde{\mathcal{L}}_{\psi}f=\barf\left[ \intS \hf(\hv') \psi(\tv|\hv') \, {\rm d} \hv'-\tf \right]$. In analogy to $\tilde{\mathcal{L}}_{\psi}$, we have  defined the reorientation operator $\hat{\mathcal{L}}_q$ as a function of $f$. However, in this case as the transition probability $q$ only depends on $\hv$, then we actually have that $\hat{\mathcal{L}}_qf=\barf\left[ q(\hv)-\hf \right]$, and the homogeneous version of~\eqref{eq:hf} would be closed with respect to $\hf$. 

To this respect, we mention here that $f$ we may always be factorized as
\begin{equation}\label{f:cond}
    f(\tv,\hv)=\rho f_c(\tv|\hv) \hf(\hv)
\end{equation}
for each $(t,x)$ fixed, that we are omitting in the notation. In~\eqref{f:cond} the function $f_c=f_c(\tv|\hv)$ is the conditional probability density of the random variable $\tV$ given $\hV=\hv$. Because of the dependence of $\psi$, notwithstanding the independence of the Bernoulli trial processes for the speed and for the direction~\eqref{ass:indep1}, the distribution of the speeds is correlated to the one of the directions. In fact, the independence only concerns the stochastic distribution of the times of reorientation and of speed-jump. The marginal of the speed can be then also written as
\begin{equation}\label{ftilde_c}
    \tf(\tv) = \intS f_c(\tv|\hv) \hf(\hv) \, {\rm d} \hv.
\end{equation}


\subsection{Insight on the stochastic trajectories}
We now want to analyze the Cauchy problem~\eqref{eq:kin_strong}-\eqref{ci} by means of the perturbed semigroup theory for linear problems. Let us first choose the appropriate Banach space. Because of the conservation property~\eqref{mass.cons}, the most natural choice is
\[
Y=L^1(\R_+\times \R^d\times \mathcal{V}),
\]
so that choosing $f_0 \in Y$, the conservation property~\eqref{mass.cons} writes
\[
||f||_Y=||f_0||_Y.
\]
Let us then rewrite~\eqref{eq:kin_strong}-\eqref{ci} as the following evolution problem in $Y$
\begin{equation}\label{probl.semi}
    \begin{cases}
        \dfrac{d}{dt}f(t) = (A+Q)f(t), \qquad t>0,\\[0.3cm]
        f(0)=f_0 \in Y,
    \end{cases}
\end{equation}
where $A$ is the operator defined by
\begin{equation}\label{def:A}
    \begin{cases}
        (Af)(x,\tv,\hv) := -v\cdot \nabla_x f(x,\tv,\hv)-(\tmu+\hmu)f(x,\tv,\hv)\\[0.3cm]
        \mathcal{D}(A) = \lbrace f \in Y | v\cdot\nabla_x f \in Y\rbrace
    \end{cases}
\end{equation}
being $\mathcal{D}(A)$ the domain of $A$. The semigroup generated by $A$, that we shall indicate as $(T(t)f)$, can be formally found by solving the problem 
\begin{equation*}
    \begin{cases}
        \partial_t f(t,x,\tv,\hv) = - v\cdot \nabla_x f(t,x,\tv,\hv) -(\tmu+\hmu) f(t,x,\tv,\hv), \\[0.3cm]
        f(0,x,\tv,\hv) = f_0(x,\tv,\hv) \in Y, \quad  (x,\tv,\hv) \in \R^d\times \mathcal{V}, \, t\ge 0,
    \end{cases}
\end{equation*}
and it is well known to be the transport semigroup (that is actually a group) defined as
\begin{equation*}
    (T(t)f)(x,\tv,\hv) := e^{-(\tmu+\hmu)} f(x-tv,\tv,\hv), \qquad f \in Y.
\end{equation*}
As to the perturbation operator $Q$, it is actually defined by the sum of two operators
\begin{equation}\label{def:Q}
Q=\hat{Q}+\tilde{Q}
\end{equation}
where
\begin{equation}\label{def:hQ}
    \begin{cases}
        (\hat{Q}f)(x,\tv,\hv) := \hmu q(\hv)\displaystyle \intR f(t,x,\tv,\hv') \, {\rm d} \hv'\\[0.3cm]
        \mathcal{D}(\hat{Q}) = Y,
    \end{cases}
\end{equation}
and 
\begin{equation}\label{def:tQ}
    \begin{cases}
        (\tilde{Q}f)(x,\tv,\hv) := \tmu \psi(\tv|\hv) \displaystyle \intS f(t,x,\tv',\hv) \, {\rm d} \tv'\\[0.3cm]
        \mathcal{D}(\tilde{Q}) = Y.
    \end{cases}
\end{equation}
It can be readily seen that 
\[
||\hat{Q}f||_Y = \int_{\R^d\times\mathcal{V}} \hmu q(\hv) \intS f(t,x,\tv,\hv') \, {\rm d} \hv' {\rm d} \tv {\rm d}\hv {\rm d}x = \hmu ||f||_Y
\]
thanks to the positivity of $q$ and to~\eqref{q.L1}, and that
\[
||\tilde{Q}f||_Y = \int_{\R^d\times\mathcal{V}} \tmu \psi(\tv|\hv) \intR f(t,x,\tv',\hv) \, {\rm d} \tv' {\rm d} \tv {\rm d}\hv {\rm d}x = \tmu ||f||_Y
\]
thanks to the positivity of $\psi$ and to~\eqref{psi.L1}. As a consequence, both $\hat{Q}f$ and $\tilde{Q}f$ are defined for each $f \in Y$ and they are linear bounded operators with $||\hat{Q}|| = \hmu$ and $||\tilde{Q}||=\tmu$. Thus, $Q$ is bounded with $||Q||=\tmu+\hmu$. 
Therefore, the hypothesis of Theorem III.1.10 in~\cite{engel_nagel} are satisfied and the integral solution of the problem~\eqref{probl.semi}-\eqref{def:A}-\eqref{def:Q}-\eqref{def:hQ}-\eqref{def:tQ} is given by the Dyson-Phillips series 
\begin{equation}\label{dyson}
\begin{aligned}[b]
    f(t) &= T(t) f_0\\[0.3cm]
    &+\int_0^t T(t-s_1)QT(s_1) f_0 \, {\rm d}s_1\\[0.3cm]
    &+\int_0^t \int_0^{s_1} T(t-s_1)QT(s_1-s_2) f_0 \, {\rm d}s_2 \, {\rm d}s_1\\[0.3cm]
    &+...,
    \end{aligned}
\end{equation}
and its limit is the integral equation for $f$ given by the Duhamel formula
\begin{equation*}
\begin{aligned}[b]
    f(t,x,\tv,\hv) &= f_0(x-tv,\tv,\hv) e^{-(\tmu+\hmu)t}\\[0.3cm] 
    &+\hmu q(\hv) \int_0^t e^{-(\tmu+\hmu)(t-s)} \rho(s,x-(t-s)v)\tf(s,x-(t-s)v,\tv) \,  {\rm d}s\\[0.3cm]
    &+\tmu \psi(\tv|\hv)\int_0^t  e^{-(\tmu+\hmu)(t-s)}  \rho(s,x-(t-s)v)\hf(s,x-(t-s)v,\hv)\,  {\rm d}s.
    \end{aligned}
\end{equation*}
The terms in~\eqref{dyson} have a physical meaning that may give an insight on the stochastic trajectories that are followed by the particles. 
The first term in~\eqref{dyson} is
\[
(T(t)f_0)(x,\tv,\hv) = e^{-(\tmu+\hmu)t} f_0(x-tv,\tv,\hv) 
\]
which takes into account the particles that have reached position $x$ with velocity vector $v$ at time $t$ without reorientation or speed-jumps: they have never modified their direction $\hv$ or speed $\tv$ since $t=0$ when they where located at $x-tv$. 
The second term refers to particles which have undergone one direction-jump and one speed-jump. It is worth writing it explicitly as
\begin{equation}\label{dyson.1}
\begin{aligned}[b]
&\int_0^t T(t-s_1)QT(s_1) f_0 \, {\rm d}s_1 = \\[0.3cm]
&+ \hmu q(\hv) \int_0^t e^{-(\tmu+\hmu)(t-\hat{s}_1)}\rho(\hat{s}_1,x-(t-\hat{s}_1)v)\tf(\hat{s}_1,x-(t-\hat{s}_1)v,\tv) \,  {\rm d}\hat{s}_1\\[0.3cm]
&+\tmu \psi(\tv|\hv)\int_0^t  e^{-(\tmu+\hmu)(t-\tilde{s}_1)}  \rho(\tilde{s}_1,x-(t-\tilde{s}_1)v)\hf(\tilde{s}_1,x-(t-\tilde{s}_1)v,\hv)\,  {\rm d}\tilde{s}_1.
\end{aligned}
\end{equation}
The second line in~\eqref{dyson.1} represents the particles that have reached position $x$ with velocity $v=\tv \hv$ after having performed exactly one direction-jump at time $\hat{s}_1$ (the term is integrated over all the possible times $\hat{s}_1$). These particles have traveled with velocity $v'= \tv \hv '$ up to $\hat{s}_1$ when they have reoriented along $\hv$, keeping $\tv$ fixed. Analogously, the third line in~\eqref{dyson.1} represents the particles that have reached position $x$ with velocity $v=\tv \hv$ after having performed exactly one speed-jump at time $\tilde{s}_1$ (the term is integrated over all the possible times $\tilde{s}_1$). These particles have traveled with velocity $v'= \tv' \hv $ up to $\tilde{s}_1$ when they have changed their speed into $\tv$, keeping $\hv$ fixed. In order to fix the ideas, we can consider for a moment the possibility $\hat{s}_1>\tilde{s}_1$. This situation is represented in Figure \ref{Fig:Duamel}.
\begin{figure}[t]
\centering
    \includegraphics[width=\textwidth]{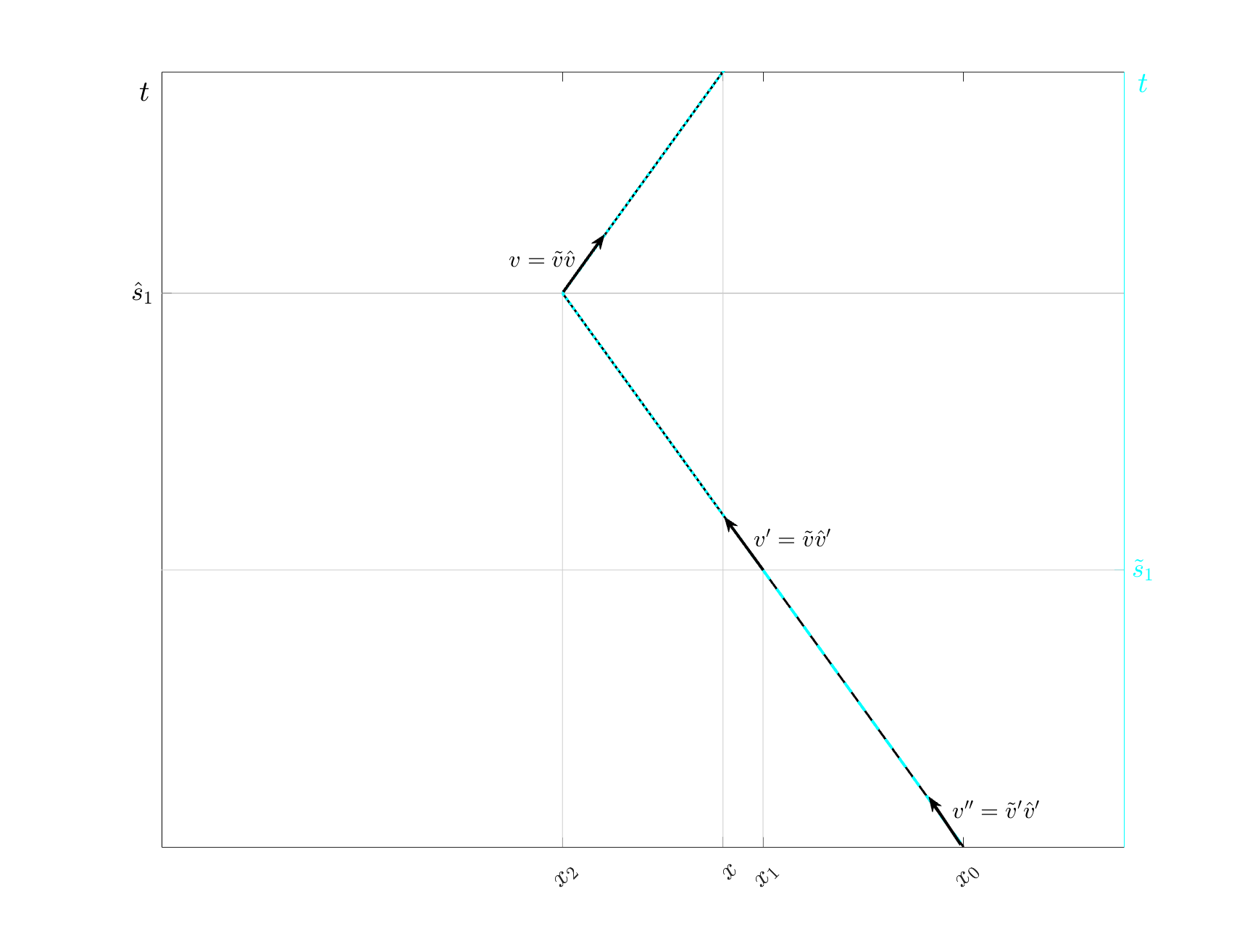}
    \caption{{\bf Graphical representation of the stochastic trajectories.} The left y-axis (black) refers to the times of the direction-jump process ($\hat{s}_1$), while the right y-axis (cyan) refers to the times of the speed-jump process ($\tilde{s}_1$). Speed changes are identified by the changes in the line style of the cyan lines (dashed or dotted lines), while direction changes are identified by the changes in orientation of the black lines. We denote the spatial point as $x_0=x-v(t-\hat{s}_1)-v'(\hat{s}_1-\tilde{s}_1)-v''\tilde{s}_1$, $x_1=x-v(t-\hat{s}_1)-v'\tilde{s}_1$, and $x_2=x-vt$.}
    \label{Fig:Duamel}
\end{figure}
At $\hat{s}_1$ the velocity $v'=\tv \hv'$ changes into $v=\tv \hv$, given that at $\tilde{s}_1$ the velocity had changed into $v'$ given the previous velocity $v''= \tv' \hv'$ held by the particle from time $0$ to $\tilde{s}_1$. Hence, those particles were located in $x-v(t-\hat{s}_1)-v'(\hat{s}_1-\tilde{s}_1)-v''\tilde{s}_1$ at time $t=0$. 

Following this interpretation, we can express the continuous-time stochastic process underlying~\eqref{eq:kin_strong} (and therefore~\eqref{eq:kin_weak}) as follows. 
Let $\lbrace \tilde{\tau}_n\rbrace_{n \in \mathbb{N}}$, $\lbrace \hat{\tau}_n\rbrace_{n \in \mathbb{N}}$ be two distinct and independent  sequences of independent identically distributed random variables in $\R_+$ with exponential law
\[
\mathbb{P}(\tilde{\tau}_n >t) = e^{-\tmu t}, \qquad \mathbb{P}(\hat{\tau}_n >t) = e^{-\hmu t}, \qquad \forall n \in \mathbb{N}.
\]
Therefore, the probability densities $p_{\tilde{\tau}}(t), p_{\hat{\tau}}(t)$ of $\tilde{\tau}_n$, $\hat{\tau}_n$ are
\[
p_{\tilde{\tau}}(t)=\tmu e^{-\tmu t}, \qquad p_{\hat{\tau}}(t)= \hmu e ^{-\hmu t}.
\]
The random variables $\tilde{\tau}_n$ and $\hat{\tau}_n$ represent the inter-arrival times in their respective renewal processes, that we shall define, respectively, as
\[
\tilde{T}_n := \sum_{i=0}^{n-1} \tilde{\tau}_i, \qquad \hat{T}_n := \sum_{i=0}^{n-1} \hat{\tau}_i.
\]
These ones are Poisson processes as the distribution of the inter-arrival times are exponential distributions.
As $\lbrace \tilde{\tau}_n\rbrace_{n \in \mathbb{N}}$, $\lbrace \hat{\tau}_n\rbrace_{n \in \mathbb{N}}$ are assumed to be independent, then $\lbrace \tilde{T}_n\rbrace_{n \in \mathbb{N}}$ and $\lbrace \hat{T}_n\rbrace_{n \in \mathbb{N}}$ are independent as well. This independence of the two renewal processes is equivalent to the one of the two trial Bernoulli processes~\eqref{ass:indep1} in the discrete-time stochastic process.

Like in the previous subsection, let us now introduce the random variables $\tV_t$, $\hV_t$ where now ${t \in \R_+}$, so that $\lbrace \tV_t \rbrace_{t\ge 0}$, $\lbrace \hV_t\rbrace_{t\ge 0}$ are continuous-time stochastic processes. Furthermore, we require that $\lbrace \tilde{\tau}_n \rbrace_{n \in \mathbb{N}}$ and $\lbrace \tV_t\rbrace_{t\ge 0}$ are independent as well as $\lbrace \hat{\tau}_n\rbrace_{n \in \mathbb{N}}$ and  $\lbrace \hV_t \rbrace_{t\ge 0}$. These two requirements have their analogous counterpart in~\eqref{ass:indep2} and~\eqref{ass:indep3}, respectively, in the discrete process. 
We here highlight the fact that $\lbrace \tV_t \rbrace_{t\ge 0}$, $\lbrace \hV_t\rbrace_{t\ge 0}$ are two piecewise deterministic Markov processes that are ruled by two different 
renewal processes  $\lbrace \tilde{T}_t \rbrace_{t\ge 0}$, $\lbrace \hat{T}_t\rbrace_{t\ge 0}$.
In order to define the stochastic trajectories in the physical space, i.e. the process $X_t$, as we have two different Poisson processes for the speed and direction, we define the renewal process for the velocity vector as
\begin{equation}\label{def:Tn}
    T_1 := \min \lbrace \tilde{T}_1, \hat{T}_1\rbrace, \quad T_n := \min\lbrace \tilde{T}_i, \hat{T}_j \rbrace - \lbrace T_i, i=0,...,n-1\rbrace  \quad {\rm for } \, \, n\ge 2.
\end{equation}
At each $T_i$ either the speed or the direction changes.
Given the initial condition $X_0=x_0$, $\tV_0=\tv_0$, $\hV_0=\hv_0$, we define $X_t,\tV_t,\hV_t$ as follows:
\begin{itemize}
    \item[-] for $0\le t <T_1$
    \[
    X_t=x_0+tv, \qquad V_t=v_0 := \tv_0\hv_0,
    \]
    so that $X_{T_1}=x_0+T_1v_0$;
    \item[-] for $T_1\le t <T_2$
    \[
    X_t=X_{T_1}+(t-T_1)V_t, \qquad V_t = \tV_t \hV_t,
    \]
    where if $T_1=\hat{T}_1$, then $\hV_t=\hV_{T_1} \sim q(\tv)$ and $\tV_t=\tv_0$, while if $T_1=\tilde{T}_1$, then $\hV_t=\hv_0$ and $\tV_t=\tV_{T_1} \sim \psi(\tv|\hv_0)$;
    \item[-] for $T_n \le t < T_{n+1}$
    \[
    X_t=X_{T_n}+(t-T_n)V_t, \qquad V_t = \tV_t \hV_t.
    \]
\end{itemize}
We remark that 
\[
\mathbb{P}(t <T_1)= \mathbb{P}(t<\tilde{T}_1 \cap t< \hat{T}_1)= \mathbb{P}(t<\tilde{T}_1) \mathbb{P}(t<\hat{T}_1)
\]
where the first equality follows from~\eqref{def:Tn}, while the second equality follows from the independence of the two renewal processes. 
Now, as $\tilde{T}_1=\tilde{\tau}_1$ and $\hat{T}_1=\hat{\tau}_1$, we may conclude that
\[
\mathbb{P}(t <T_1)=e^{-(\tmu+\hmu)t}
\]
which is the factor appearing in the first term in~\eqref{dyson}, that is the one related to particles which have not undergone a jump process, neither the direction or the speed one.

\section{Investigation of the operators}\label{sec:operators}

In this section we study some of the properties of the new operators that define the transport equation~\eqref{eq:boltzmann}-\eqref{def:L}. First, we analyze the kernels and the stationary states. After defining the appropriate Hilbert spaces, we determine the pseudo-inverse operators, with the purpose of deriving macroscopic limits in the following section. Building upon the obtained results, we establish some properties for the entropy decay to equilibrium in the spatially homogeneous case.

\subsection{Kernels and stationary states}

As already mentioned, the linear transport equation~\eqref{eq:boltzmann}-\eqref{def:L}, has only a formal similarity with the linear Boltzmann equation, because the operator $\mathcal{L}$ is defined by the sum of the operators $\tilde{\mathcal{L}}$ and $\hat{\mathcal{L}}$. As such, $\mathcal{L}$ is not in the form~\eqref{operator.proto} as in~\eqref{kin.eq:M}, where $M$ is both the turning kernel and the equilibrium, being $\ker(L)={\rm span}\lbrace M\rbrace$. This poses many technical issues, the first one being the determination of the kernel of $\mathcal{L}$ and of an explicit kinetic equilibrium. To fill this gap, we establish the following result on the kernel of the operator $\mathcal{L}$.

\begin{proposition}\label{prop.1}
Let us consider $\psi$ satisfying~\eqref{norm.psi},\eqref{psi.L1}, $q$ satisfying~\eqref{norm.q},\eqref{q.L1}, and the operator $\mathcal{L}$ defined by~\eqref{def:L}. We have that
\[
\ker(\mathcal{L})= {\rm span}\lbrace T\rbrace,
\]
where
\begin{equation}\label{def:T}
    T: \mathcal{V}\to \R_+ \qquad T(v,\hv):= q(\hv) \psi_q^c(\tv|\hv),
\end{equation}
being
\begin{equation}\label{def:psiqc}
    \psi_q^c(\cdot|\hv):\R_+ \to\R_+, \qquad \psi_q^c(\tv|\hv) := \dfrac{\hmu \psi_q(\tv)+\tmu\psi(\tv|\hv)}{\hmu+\tmu}, \qquad \forall \hv \in \mathbb{S}^{d-1},
\end{equation}
where
\begin{equation*}\label{def:psiq}
    \psi_q:\R_+\to \R_+, \qquad \psi_q(\tv):=\intS \psi(\tv|\hv) \, q(\hv) \, {\rm d} \hv.
\end{equation*}
Moreover
\begin{equation}\label{norm.psiq}
\psi_q \in L^1(\R_+), \qquad ||\psi_q||_{L^1(\R_+)}=1,    
\end{equation}

\begin{equation}\label{norm.psiqc}
 \psi_q^c(\cdot| \hv) \in L^1(\R_+), \,  \qquad ||\psi_q^c||_{L^1(\R_+)}=1, \quad \forall \hv \in \mathbb{S}^{d-1},
\end{equation}
and
\begin{equation}\label{norm.T}
T\in L^1(\mathcal{V}), \qquad ||T||_{L^1(\mathcal{V})}=1.    
\end{equation}
\end{proposition}
\begin{proof}
 By definition~\eqref{def:L}, we have that
 \begin{equation}\label{proof1}
 \varphi \in \ker(\mathcal{L}) \quad \text{iff} \quad \varphi(\tv,\hv) = \rho_{\varphi}\dfrac{\hmu \tphi(\tv)q(\hv)+\tmu\hphi (\hv) \psi(\tv|\hv)}{\hmu+\tmu}.
 \end{equation}
 Therefore, if $\varphi \in \ker(\mathcal{L})$, direct computations using~\eqref{def:bthphi} show that
 \[
 \hphi(\hv)= q(\hv), \qquad \tphi(\tv) = \intV \hphi(\hv) \psi(\tv|\hv) \, {\rm d} \hv.
 \]
 The latter implies that
\begin{equation}\label{proof2}
\hphi(\hv)= q(\hv), \qquad \tphi(\tv) = \psi_q(\tv) := \intV  \psi(\tv|\hv) q(\hv) \, {\rm d} \hv.
\end{equation}
Taking into account the positivity of $q$ and $\psi$, we have that $\psi_q$ is positive. Thanks to~\eqref{norm.psi},\eqref{norm.q}, and~\eqref{psi.L1},\eqref{q.L1} for inverting the order of integration, we have that 
\[
||\psi_q||_{L^1(\R_+)}:= \intR \psi_q(\tv) \, {\rm d} \tv = \intS q(\hv) \intR \psi(\tv|\hv) \, {\rm d} \tv \, {\rm d} \hv =1,
\]
i.e., $\psi_q \in L^1(\R_+)$ and~\eqref{norm.psiq} is verified.
Therefore, substituting~\eqref{proof2} in \eqref{proof1}, we have that
\[
\varphi(\tv,\hv)=\rho_{\varphi} T(\tv, \hv),
\]
being $T$ defined in~\eqref{def:T}, with $\psi_q^c$ defined in~\eqref{def:psiqc}. It is clear that $\psi_q^c$ and $T$ are positive thanks to the positivity of $\psi$ and $q$. Thanks to~\eqref{norm.q},\eqref{norm.psiq}, we have that~\eqref{norm.psiqc} and then~\eqref{norm.T} are verified as well. In conclusion, $\varphi \in {\rm span} \lbrace T \rbrace$. 

Conversely, if $\varphi \in \text{span}\lbrace T \rbrace$, then $\varphi=\rho_{\varphi} T$ for~\eqref{norm.T}. Direct computations show that $\mathcal{L}\rho_\varphi T=0$, and therefore $\varphi \in \ker(\mathcal{L})$.
\end{proof}

It is now worth to investigate the kernels of the marginal operators $\hat{\mathcal{L}}_q$, $\tilde{\mathcal{L}}_{\psi}$ and their relation to the kernel of $\mathcal{L}$.
\begin{proposition} 
Let us consider $\psi$ satisfying~\eqref{norm.psi},\eqref{psi.L1}, $q$ satisfying~\eqref{norm.q},\eqref{q.L1}, and the operators $\hat{\mathcal{L}}_q$, $\tilde{\mathcal{L}}_{\psi}$ defined by~\eqref{def:Lq}, \eqref{def:Lpsi}. We have that
\begin{equation}\label{ker:Lq}
\begin{split}
  \ker (\hat{\mathcal{L}}_q) &= \Big\lbrace \varphi \in L^2(\mathcal{V}) \, : \, \hphi(\hv) =  q(\hv)\Big\rbrace \\
  &= \left\lbrace \varphi \in L^2(\mathcal{V}) \, : \, \intR \varphi(\tv,\hv) \, {\rm d} \tv \in {\rm span}\lbrace q\rbrace\right\rbrace,
  \end{split}
\end{equation}
and
\begin{equation}\label{ker:Lpsi}
\begin{aligned}[b]
  \ker (\tilde{\mathcal{L}}_\psi) &= \left\lbrace \varphi \in L^2(\mathcal{V}) \, : \, \tphi(\tv) =  \intV\psi(\tv|\hv) \, \hphi(\hv) \, {\rm d} \hv \right\rbrace\\ 
  &\supseteq \left\lbrace \varphi \in L^2(\mathcal{V}) \, : \, \intS \varphi(\tv,\hv) \, {\rm d} \hv \in {\rm span}\lbrace \psi_q \rbrace\right\rbrace. 
  \end{aligned}
\end{equation}
Moreover,
\begin{equation}\label{ker:Lqpsi}
         \ker (\mathcal{L}) \subseteq \ker (\hat{\mathcal{L}}_q)  \cap \ker (\tilde{\mathcal{L}}_\psi).
\end{equation}
\end{proposition}
\begin{proof}
It is straightforward to verify~\eqref{ker:Lq}-\eqref{ker:Lpsi}.
If $\varphi \in \ker(\mathcal{L})$, then for Proposition~\ref{prop.1} $\varphi=\rho_\varphi T$. Therefore $\hphi(\hv)=q(\hv)$ so that $\varphi \in\ker(\hat{\mathcal{L}}_q)$, but also $\tphi(\tv)=\psi_q(\tv)$ so that $\varphi \in \ker(\tilde{\mathcal{L}}_\psi)$. 
\end{proof}
\begin{remark}
If $\varphi \in \ker (\hat{\mathcal{L}}_q)  \cap \ker (\tilde{\mathcal{L}}_\psi)$, then $\hphi=q$ and $\tphi=\psi_q$. However, this does not allow to conclude that $\varphi=\rho_\varphi T$. In fact, for each $\lambda, \eta > 0,$ ${\varphi=\rho_\varphi\dfrac{\lambda q(\hv) \psi(\tv|\hv)+ \eta \psi_q(\tv) q(\hv)}{\lambda +\eta} \in \ker (\hat{\mathcal{L}}_q)  \cap \ker (\tilde{\mathcal{L}}_\psi)}$, but it belongs to $\ker(\mathcal{L})$ only if $\lambda=\tmu, \eta=\hmu$.  This is why in general the inclusion in~\eqref{ker:Lqpsi} is not an equality.
\end{remark}

Thanks to Proposition~\ref{prop.1}, being $T$ spatially homogeneous, it is possible to state that the stationary state and equilibrium of~\eqref{eq:boltzmann} is
\begin{equation*}
    f^{\infty}(\tv,\hv) := \rho T(\tv, \hv), \qquad T(\tv,\hv) := q(\hv) \psi_q^c(\tv|\hv),
\end{equation*}
where
\begin{equation}\label{hf.infty}
    \hf^\infty(\hv) =q(\hv)
\end{equation}
thanks to~\eqref{ker:Lq}.
In view of~\eqref{f:cond}, we have that 
\begin{equation}\label{fc.infty}
    f_c^{\infty}(\tv|\hv) = \psi_q^c(\tv|\hv).
\end{equation}
As a consequence, thanks to~\eqref{ftilde_c}, it is possible to find the stationary state of the marginal $\tf$, that is
\begin{equation}\label{tf.infty}
    \tf^\infty(\tv) = \psi_q(\tv).
\end{equation}
We can also remark that the inclusion in~\eqref{ker:Lqpsi} is due do the dependence of the process $\lbrace \tV_t \rbrace_t$ on $\lbrace\hV_t \rbrace_t$ through $\psi(\tv|\hv)$. In fact, if the two processes are completely uncorrelated, i.e., $\psi=\psi(\tv)$ does not depend on the direction $\hv$, then we have that
\begin{equation}\label{psiv.1}
    \psi(\tv) = \psi_q(\tv) = \psi_q^c(\tv).
\end{equation}
The latter implies that
\begin{equation}\label{psiv.2}
    \tf^\infty(\tv) =f^\infty_c(\tv) = \psi(\tv),
\end{equation}
which means that at the stationary state the two random variables $\tV_t, \hV_t$ are independent, i.e.,
\begin{equation}\label{f.indep}
f^\infty(\tv, \hv) = \rho\tf^\infty(\tv) \hf^\infty (\hv).
\end{equation}

Coming again to the issue of the gain term of $\mathcal{L}$, the difference with respect to the classical case~\eqref{kin.eq:M}, is that it is not in the form $\rho T$, being $T$ the equilibrium. In order to establish a connection, we can remark that $\mathcal{L}$ can be written as
\begin{equation}\label{L.tau}
    \mathcal{L}f = \mathcal{T}(f)-(\tmu+\hmu) f, \qquad \mathcal{T}(f) := \rho\big(\hmu \tf q+\tmu \hf \psi\big),
\end{equation}
where $\mathcal{T}(f)$ is not $(\tmu +\hmu) T$.
Manipulating $\mathcal{T}(f)$, we have that
\begin{equation}\label{Tau}
\mathcal{T}(f) = \hmu (\tf \pm \psi_q)\rho q + \tmu (\hf \pm q) \rho\psi = - \mathcal{L}_{\psi_q}(\rho \tf) q -  \mathcal{L}_q(\rho \hf) \psi + (\hmu+\tmu)\rho T
\end{equation}
where
\begin{equation}\label{def:Lq.marg}
\mathcal{L}_{q}:L^2(\mathbb{S}^{d-1}) \to L^2(\mathbb{S}^{d-1}), \qquad \mathcal{L}_{q}(\rho \hf) = \hmu\rho (q -\hf),
\end{equation}
and
\begin{equation*}
\mathcal{L}_{\psi_q}: L^2(\R_+) \to L^2(\R_+), \qquad \mathcal{L}_{\psi_q}(\rho \tf) = \tmu\rho (\psi_q -\tf).
\end{equation*}
Differently with respect to $\hat{\mathcal{L}}_q, \tilde{\mathcal{L}}_{\psi}$, the operators $\mathcal{L}_{q},\mathcal{L}_{\psi_q}$ map the same functional spaces ($L^2(\mathbb{S}^{d-1})$ and $ L^2(\R_+)$, respectively), and therefore  $\ker(\mathcal{L}_q)={\rm span}\lbrace q \rbrace$, while $\ker(\mathcal{L}_{\psi_q})={\rm span}\lbrace \psi_q \rbrace$. This mirrors the fact that the stationary states of the marginals are $q$ (see~\eqref{hf.infty}) and $\psi_q$ (see~\eqref{tf.infty}).
As a consequence
\begin{equation}\label{def:L.T}
\mathcal{L}f =  \mathcal{L}_Tf - \mathcal{L}_{\psi_q}(\rho \tf) q - \mathcal{L}_q(\rho \hf) \psi  
\end{equation}
being
\[
\mathcal{L}_Tf= (\hmu+\tmu)\big(\rho T -f\big)
\]
i.e. the reorientation process described by $\mathcal{L}$ is a linear combination of the reorientation process $\mathcal{L}_T$ dictated by the equilibrium $T$ and the opposite of the marginal processes $-\mathcal{L}_q,-\mathcal{L}_{\psi_q}$. It is worth mentioning that, while $\mathcal{L}_q$ actually rules the evolution of $\hf$ given by~\eqref{eq:hf.bis}, the operator $\mathcal{L}_{\psi_q}$ does not explicitly appear in the evolution of the marginal $\tf$, i.e.,~\eqref{eq:tf.bis}, which is ruled by $\psi$. Only when $\psi$ does not depend on $\hv$ and~\eqref{psiv.1} holds true, then $\tilde{\mathcal{L}}_{\psi}(\rho \tf)=\mathcal{L}_{\psi_q} (\rho \tf)$. 


\subsection{Hilbert spaces and pseudo-inverse operators}
Building on the results of the previous subsection, we now want to define the appropriate Hilbert spaces
with the purpose of defining the pseudo-inverse operators, being the final aim the derivation of macroscopic limits. 

In view of the fact that $T$ is the equilibrium distribution as it is the generator of $\ker(\mathcal{L})$, it is natural to define the following scalar product on $L^2(\mathcal{V})$
\begin{equation}\label{sc.pr.T}
    \eta, \xi \in L^2(\mathcal{V}), \quad \ave{\eta,\xi}_T:=\intV \eta(\tv,\hv) \xi(\tv,\hv) T(\tv,\hv)^{-1} \, {\rm d} \tv \, {\rm d} \hv.
\end{equation}
The space $L^2_T(\mathcal{V}):=(L^2(\mathcal{V}),T^{-1})$ endowed with the scalar product~\eqref{sc.pr.T} is a Hilbert space. The Fredholm alternative for $\mathcal{L}_T$ allows to state that
\[
L^2_T= \text{span}\left\lbrace T\right\rbrace \oplus \text{span}\left\lbrace T\right\rbrace^{\bot}
\]
where
\begin{equation*}
 \text{span}\left\lbrace T\right\rbrace^{\bot}=\left\lbrace \varphi: 0=\ave{\varphi,T}_T=\intV \varphi(\tv,\hv) \, {\rm d} \tv \, {\rm d} \hv \right\rbrace.   
\end{equation*}
Analogously, as the equilibrium of the marginal $\hf$ is $q$ and $\ker(\mathcal{L}_q)=\text{span}\lbrace q \rbrace$, we introduce the scalar product on $L^2(\mathbb{S}^{d-1})$ as
\begin{equation*}\label{sc.pr.q}
    \eta, \xi \in L^2(\mathbb{S}^{d-1}), \quad \ave{\eta,\xi}_q:=\intS \eta(\hv) \xi(\hv) q(\hv)^{-1}  \, {\rm d} \hv,
\end{equation*}
which defines a Hilbert space $L^2_q(\mathbb{S}^{d-1}):=(L^2(\mathbb{S}^{d-1}),q^{-1})$. 
Here we have that the Fredholm alternative for $\mathcal{L}_q$ allows to state that
\[
L^2_q(\mathbb{S}^{d-1}) = \text{span}\lbrace q \rbrace \oplus \text{span}\lbrace q \rbrace^\bot, 
\]
where
\begin{equation*}
 \text{span}\left\lbrace q\right\rbrace^{\bot}=\left\lbrace \eta\in L^2(\mathbb{S}^{d-1}) : 0=\ave{\eta,q}_q=\intS \eta(\hv) \, {\rm d} \hv \right\rbrace.   
\end{equation*}
Then, the classical theory applies and the operator $\mathcal{L}_q$ can be inverted in $L^2_q(\mathbb{S}^{d-1})$ on the orthogonal to its kernel, i.e,
\begin{equation*}
\begin{aligned}[b]
    \mathcal{L}_q^{-1}:  \text{span}\left\lbrace q\right\rbrace^{\bot} &\to \text{span}\left\lbrace q\right\rbrace^{\bot}\\
    \eta &\to -\dfrac{1}{\hmu}\eta.
    \end{aligned}
\end{equation*}
Now, $\hat{\mathcal{L}}_q$, whose kernel is given by~\eqref{ker:Lq}, is defined on $L^2_T(\mathcal{V})$, but maps into $L^2_q(\mathbb{S}^{d-1})$, thus it must be inverted on $ \text{span}\left\lbrace q\right\rbrace^{\bot}$. Noticing that
\begin{equation*}\label{qbot}
    \varphi \in \text{span}\lbrace T \rbrace^{\bot} \Rightarrow \intR \varphi {\rm d}\tv \in \text{span}\lbrace q \rbrace^{\bot},
\end{equation*}
we have that
\begin{equation}\label{Lq.inv}
\begin{aligned}[b]
    \hat{\mathcal{L}}_q^{-1}:  \text{span}\left\lbrace q\right\rbrace^{\bot} &\to \text{span}\left\lbrace T\right\rbrace^{\bot}\\
    \eta &\to -\dfrac{1}{\hmu}\varphi \quad \text{s.t.} \quad \intR \varphi(\tv,\hv) {\rm d} \tv = \eta(\hv),
    \end{aligned}
\end{equation}
i.e., $\rho_\varphi \hphi=\eta.$
In the same spirit, we also define a scalar product on $L^2(\R_+)$ that will be connected to $\psi_q(\tv)$, which is the equilibrium of the marginal $\tf$, and that generates $\ker(\mathcal{L}_{\psi_q})$, but not $\ker(\tilde{\mathcal{L}}_{\psi})$. Specifically, we set
\begin{equation}\label{sc.pr.psi}
    \eta, \xi \in L^2(\R_+), \quad \ave{\eta,\xi}_{\psi_q}:=\intR \eta(\tv) \xi(\tv) \psi_q(\tv)^{-1}  \, {\rm d} \tv,
\end{equation}
so that $L^2_{\psi_q}(\R_+):=(L^2(\R_+),\psi_q^{-1})$ endowed with the scalar product~\eqref{sc.pr.psi} is a Hilbert space. Here we have again that the Fredholm alternative for $\mathcal{L}_{\psi_q}$ allows to state that 
\[
L^2_{\psi_q} = \text{span}\left\lbrace \psi_q\right\rbrace \oplus \text{span}\left\lbrace \psi_q\right\rbrace^{\bot}, 
\]
where
\begin{equation*}
 \text{span}\left\lbrace \psi_q\right\rbrace^{\bot}=\left\lbrace \xi \in L^2(\R_+): 0=\ave{\xi,\psi_q}_{\psi_q}=\intR \xi(\tv) \, {\rm d} \tv \right\rbrace.  
\end{equation*}
The operator $\mathcal{L}_{\psi_q}$ can be inverted in $L^2_{\psi_q}(\R_+)$ on the orthogonal to its kernel, i.e, $\text{span}\left\lbrace \psi_q\right\rbrace^{\bot}$, and we have that
\begin{equation*}
\begin{aligned}[b]
    \mathcal{L}_{\psi_q}^{-1}:  \text{span}\left\lbrace \psi_q\right\rbrace^{\bot} &\to \text{span}\left\lbrace \psi_q\right\rbrace^{\bot}\\
    \eta &\to -\dfrac{1}{\tmu}\eta.
    \end{aligned}
\end{equation*}
Now, as we actually shall need to invert $\tilde{\mathcal{L}}_\psi$, analogously to the previous case $\hat{\mathcal{L}}_q$, we remark that it maps two different functional spaces, $L^2_T(\mathcal{V})$ into $L^2_{\psi_q}(\R_+)$, and its kernel is described by~\eqref{ker:Lpsi}, and we note that
\begin{equation*}\label{psibot}
    \varphi \in \text{span}\lbrace T \rbrace^{\bot} \Rightarrow  \intS \varphi {\rm d}\hv \in \text{span}\lbrace \psi_q \rbrace^{\bot}.
\end{equation*}
However, we now observe that, differently from the previous case $\hat{\mathcal{L}}_q$, even if $\intS \varphi {\rm d}\hv \in \text{span}\lbrace \psi_q \rbrace^{\bot}$, the gain term of $\tilde{\mathcal{L}}_\psi(\varphi)$ does not vanish. As a consequence, the pseudo-inverse operator reads
\begin{equation}\label{Lpsi.inv}
\begin{aligned}[b]
    \tilde{\mathcal{L}}_{\psi}^{-1}:  \text{span}\left\lbrace \psi_q\right\rbrace^{\bot} &\to \text{span}\left\lbrace T\right\rbrace^{\bot}\\
    \eta &\to -\dfrac{1}{\tmu}\varphi\\ &\text{s.t.} \quad \eta(\tv)=  \intS \varphi(\tv,\hv) \, {\rm d} \hv -\intR \varphi (\tv', \hv) \intS \psi(\tv|\hv) \, {\rm d} \hv\, {\rm d} \tv',
    \end{aligned}
\end{equation}
i.e., $\eta(\tv)= \rho_\varphi \tphi(\tv)-\rho_\varphi \intS\hphi(\hv) \psi(\tv|\hv) \, {\rm d} \hv $.

In the same spirit, $\mathcal{L}$ can be inverted on the orthogonal to its kernel, but, as it is in the form~\eqref{L.tau}, i.e., its gain term $\mathcal{T}$ is not in the form $\rho T$, then it does not vanish. In fact, even though $\varphi \in \text{span}\left\lbrace T\right\rbrace^{\bot}$, the quantities $\rho_\varphi\tphi, \rho_\varphi\hphi$ are not zero. It is their integral over $\R_+$ and $\mathbb{S}^{d-1}$, respectively, which vanishes.
As a consequence the pseudo-inverse of $\mathcal{L}$ defined in~\eqref{def:L} is
\begin{equation}\label{pseudo.inv.L}
\begin{split}
    \mathcal{L}^{-1}:  \text{span}\left\lbrace T\right\rbrace^{\bot} &\to \text{span}\left\lbrace T\right\rbrace^{\bot}\\
    \eta &\to \dfrac{1}{(\hmu+\tmu)}\Big(-\eta+\mathcal{T}(\varphi) \Big),
    \\ &\text{s.t.}\quad  \eta(\tv,\hv)=\mathcal{T}(\varphi)-(\tmu+\hmu)\varphi(\tv,\hv)
    \end{split}
\end{equation}
where $\mathcal{T}(\varphi)=\rho_\varphi(\hmu\tphi q +\tmu\hphi \psi)$. 

We remark that, as $f \in L^2_T(\mathcal{V})$, then it is immediate to verify that $\hf \in L^2_q(\mathbb{S}^{d-1})$. In fact, thanks to Jensen's inequality applied with respect to the measure $\psi_q^c(\cdot|\hv)$ for each $\hv \in\mathbb{S}^{d-1}$ fixed, we have
\begin{equation}\label{norm.hf}
\begin{split}
   \rho^2 ||\hf||^2_{L^2_q}= ||(\rho \hf)||^2_{L^2_q} &=\intS (\rho \hf(\hv))^2 q(\hv)^{-1} \, {\rm d} \hv\\
   &=\intS \Big(\intR f(\tv,\hv)\dfrac{\psi_q^c(\tv|\hv)}{\psi_q^c(\tv|\hv)} \, {\rm d} \tv\Big)^2 q(\hv)^{-1} \, {\rm d} \hv\\
   &\le \intS\intR f^2(\tv,\hv) \dfrac{1}{\psi_q^c(\tv|\hv)q(\hv)} \, {\rm d} \hv \, {\rm d} \tv  = ||f||^2_{L^2_T}. 
   \end{split}
\end{equation}
Conversely, due to the definition of $T$ which depends on $\psi_q^c$ and not on $\psi_q$, it is not immediate to verify that $\tf \in L^2_{\psi_q}$. However, assuming that
\begin{equation}\label{psiq.Linfty}
    \dfrac{\psi_q^c}{\psi_q} \in L^\infty(\mathcal{V}),
\end{equation}
thanks to Holder's inequality, we have that 
\begin{equation}\label{norm.tf}
\begin{split}
   \rho^2 ||\tf||^2_{L^2_{\psi_q}}= ||(\rho \tf)||^2_{L^2_{\psi_q}} &=\intR (\rho \tf(\tv))^2 \psi_q(\tv)^{-1} \, {\rm d} \tv \\
   &=\intR \Big(\intS f(\tv,\hv)\dfrac{q(\hv)}{q(\hv)} \, {\rm d} \hv\Big)^2  \psi_q(\tv)^{-1}\, {\rm d} \tv\\
   &\le \intR\intS f^2(\tv,\hv) \dfrac{1}{\psi_q(\tv)q(\hv)} \dfrac{\psi_q^c(\tv|\hv)}{\psi_q^c(\tv|\hv)} \, {\rm d} \tv \, {\rm d} \hv \\
   &\le ||\dfrac{\psi_q^c}{\psi_q}||_{L^\infty} ||f||^2_{L^2_T}
   \end{split} 
\end{equation}
i.e. $\tf \in L^2_{\psi_q}$.

Eventually, we prove the following boundedness property for $\mathcal{L}$.
\begin{proposition} Let us assume that conditions~\eqref{norm.psi},\eqref{psi.L1},\eqref{norm.q},\eqref{q.L1},\eqref{psiq.Linfty} hold true. Then, the linear operator $\mathcal{L}$ defined in~\eqref{def:L} is bounded in $L^2_T(\mathcal{V})$, i.e., $\exists C>0$ such that $\forall f \in L^2_T(\mathcal{V})$, $||\mathcal{L}f||^2_{L^2_T} \le C ||f||^2_{L^2_T}$.
\end{proposition}

\begin{proof}
Dropping the dependencies, we have that
 \begin{equation}\label{proof.bound1}
 \begin{split}
 || \mathcal{L}f||^2_{L^2_T} &= \intV \big(\mathcal{T}(f) -(\tmu +\hmu)f \big)^2 T^{-1} \, {\rm d} \tv \, {\rm d}\hv\\ &= \intV \mathcal{T}(f)^2 T^{-1} \, {\rm d} \tv \, {\rm d}\hv + (\tmu +\hmu)^2\intV f^2 T^{-1}\, {\rm d} \tv \, {\rm d}\hv\\ &\quad-2  (\tmu +\hmu) \intV \mathcal{T}(f) f T^{-1}   \, {\rm d} \tv \, {\rm d}\hv\\
 &\le \intV \mathcal{T}(f)^2 T^{-1} \, {\rm d} \tv \, {\rm d}\hv + (\tmu +\hmu)^2\intV f^2 T^{-1}\, {\rm d} \tv \, {\rm d}\hv\\ &\quad+2(\tmu +\hmu) \intV \mathcal{T}(f) f T^{-1} \, {\rm d} \tv \, {\rm d}\hv.
 \end{split}
 \end{equation} 
 Now, in view of~\eqref{Tau}, we have that
 \begin{equation}\label{proof.bound2}
 \begin{split}
 \intV \mathcal{T}(f)^2 T^{-1} \, {\rm d} \tv \, {\rm d}\hv&= \intV \left(\mathcal{L}_{\psi_q}(\rho \tf)^2 q^2 +\mathcal{L}_q(\rho \hf)^2 \psi^2 +(\tmu+\hmu)^2 \rho^2T^2 \right)T^{-1} \, {\rm d} \tv \, {\rm d}\hv\\
 &\quad +2\intV \left( \mathcal{L}_{\psi_q}(\rho\tf)\mathcal{L}_q(\rho \hf)\dfrac{\psi}{\psi_q^c}-(\hmu+\tmu) \left(\mathcal{L}_{\psi_q}(\rho\tf)q + \mathcal{L}_q(\rho\hf)\psi\right)\rho\right)\, {\rm d} \tv \, {\rm d}\hv \\
 &\le  \intV \left(\dfrac{\mathcal{L}_{\psi_q}(\rho \tf)^2}{\psi_q^c}q + \dfrac{\mathcal{L}_q(\rho \hf)^2}{q} \dfrac{\psi^2}{\psi_q^c} +(\tmu+\hmu)^2\rho^2 T \right) \, {\rm d} \tv \, {\rm d}\hv\\
 &\quad+2\rho^2\tmu\hmu\intV \Big(\psi_qq +\tf \hf\Big)\dfrac{\psi}{\psi_q^c} \, {\rm d} \tv \, {\rm d}\hv \\
\ &\le \intV \left( \dfrac{\mathcal{L}_{\psi_q}(\rho \tf)^2}{\psi_q}q\dfrac{\tmu+\hmu}{\hmu}  +\dfrac{\mathcal{L}_q(\rho \hf)^2}{q} \psi_q^c\dfrac{\big(\tmu+\hmu\big)^2}{\tmu^2} +(\tmu+\hmu)^2\rho^2T \right) \, {\rm d} \tv \, {\rm d}\hv\\
&\quad+2\rho^2(\tmu+\hmu)\hmu\intV \Big(\psi_qq +\tf \hf\Big)\, {\rm d} \tv \, {\rm d}\hv.
 \end{split}
 \end{equation}
 In~\eqref{proof.bound2}, the first inequality is justified by $\intS\mathcal{L}_q(\rho \hf)\, {\rm d}\hv = \intR\mathcal{L}_{\psi_q}(\rho\tf)\, {\rm d}\tv = 0$ and~\eqref{norm.psi}, and by explicit computations of the product of the two operators $\mathcal{L}_q,\mathcal{L}_{\psi_q}$ where we neglect the negative terms; the second inequality in~\eqref{proof.bound2} is justified by, in the order (for the first, second and fourth term, respectively)
 \[
\dfrac{1}{\psi_q^c} < \dfrac{\tmu+\hmu}{\hmu}\dfrac{1}{\psi_q}, \qquad \tmu \psi \le \psi_q^c (\tmu+\hmu), \qquad \dfrac{\psi}{\psi_q^c}<{\dfrac{\hmu+\tmu}{\tmu}},
 \]
 which follow from~\eqref{def:psiqc} and its positivity.
 Now we remark that $\intS \mathcal{L}_{q}(\rho\hf)^2q^{-1}\,{\rm d}\hv= ||\mathcal{L}_q(\rho\hf)||^2_{L^2_q}$ and $\intR \mathcal{L}_{\psi_q}(\rho\tf)^2\psi_q^{-1}\,{\rm d}\tv= ||\mathcal{L}_{\psi_q}(\rho\tf)||^2_{L^2_{\psi_q}}$. Using~\eqref{norm.marg}, we can write the following bounds
 \[
 || \mathcal{L}_q(\rho \hf)||^2_{L^2_q} \le \hmu^2\intS \rho^2(q^2 +\hf^2 +2 \hf q)q^{-1} \, {\rm d}\hv =\hmu^2\rho^2(|| \hf ||^2_{L^2_q} + 3\big),
 \]
 and
 \[
 ||\mathcal{L}_{\psi_q}(\rho\tf)||^2_{L^2_{\psi_q}} \le \tmu^2 \rho^2 \left(|| \tf ||^2_{L^2_{\psi_q}}+3\right).
 \]
Therefore, observing that the last term in~\eqref{proof.bound2} is $4\rho^2(\tmu+\hmu)\hmu$ thanks to~\eqref{norm.psi},\eqref{norm.q},\eqref{norm.marg}, we have that
 \begin{equation}\label{proof.bound3}
  \intV \mathcal{T}(f)^2 T^{-1} \, {\rm d} \tv \, {\rm d}\hv \le \rho^2 \left(   \dfrac{\tmu^2\big(\tmu+\hmu\big)}{\hmu} (||  \tf ||^2_{L^2_{\psi_q}}+3) +\dfrac{\hmu^2\big(\tmu+\hmu\big)^2}{\tmu^2} (|| \hf ||^2_{L^2_{q}}+3) +(\tmu+\hmu)^2+4(\tmu+\hmu)\hmu\right),
 \end{equation}
where for the first and second terms we have used~\eqref{norm.q} and~\eqref{norm.psiqc}, while~\eqref{norm.T} is used for the third term.

It is immediate to verify that, for the Jensen inequality with respect to the measure $T$ on $\mathcal{V}$, we have that
\begin{equation}\label{jensen.T}
 \rho^2 \le ||f||^2_{L^2_T}.
 \end{equation}
 Therefore, in view of~\eqref{proof.bound3} and using~\eqref{norm.hf}-\eqref{norm.tf}-\eqref{jensen.T},  the first term in the inequality in~\eqref{proof.bound1} can be bounded from above as
 \[
  \intV \mathcal{T}(f)^2 T^{-1} \, {\rm d} \tv \, {\rm d}\hv \le ||f||^2_{L^2_T} \left(   \dfrac{\tmu^2\big(\tmu+\hmu\big)}{\hmu} \left(|| \dfrac{\psi_q^c}{\psi_q}||^2_{L^\infty}+3\right) +\dfrac{4\hmu^2\big(\tmu+\hmu\big)^2}{\tmu^2}  +(\tmu+\hmu)^2+4(\tmu+\hmu)\hmu\right).
 \]
 The third term in the inequality in~\eqref{proof.bound1} can be bounded as
 \[
 \begin{split}
 2  (\tmu +\hmu) \intV \mathcal{T}(f) f T^{-1}   \, {\rm d} \tv \, {\rm d}\hv &\le 2 \rho (\tmu +\hmu) \left(\intV \hmu f \tf \psi_q(\tv)^{-1}\dfrac{\tmu+\hmu}{\hmu}   \, {\rm d} \tv \, {\rm d}\hv + \intV \tmu f \hf q(\hv)^{-1}\dfrac{\tmu+\hmu}{\tmu}   \, {\rm d} \tv \, {\rm d}\hv\right)\\
 &= 2\rho^2 (\tmu+\hmu)^2 \big(|| \tf ||^2_{L^2_{\psi_q}}+|| \hf ||^2_{L^2_{q}}\big)
 \end{split}
 \]
 where in the first inequality we have used $(\psi_q^{c})^{-1}\le \psi_q^{-1}\frac{\tmu+\hmu}{\hmu}$ and $\psi/\psi_q^c \le \dfrac{\tmu+\hmu}{\tmu}$, while in the second equality we have used~\eqref{def:hf}-\eqref{def:tf}.
  In conclusion, using~\eqref{norm.hf}-\eqref{norm.tf}-\eqref{jensen.T} where needed, we may conclude that
 \[
  || \mathcal{L}f||^2_{L^2_T} \le C ||f||^2_{L^2_T}
 \]
 where
 \[
C = \left(\tmu+\hmu\right)\left(  \dfrac{\tmu^2}{\hmu}\left(|| \dfrac{\psi_q^c}{\psi_q} ||_{L^\infty}+3\right) + \dfrac{4\hmu^2\left(\tmu+\hmu\right)}{\tmu^2} +2\left(\tmu+\hmu\right)+4\hmu+2\left(\tmu+\hmu\right)\left(|| \dfrac{\psi_q^c}{\psi_q} ||_{L^\infty}+1\right)\right)
 \] is a positive finite constant. 
\end{proof}

\subsection{Relaxation to Equilibrium in the Spatially Homogeneous Case}\label{sec:entropy}
The introduced Hilbert spaces $L^2_T(\mathcal{V}),L^2_{q}(\mathbb{S}^{d-1}),L^2_{\psi_q}(\R_+)$ provide the correct framework for studying the decay of entropy in the spatially homogeneous case, i.e. when $f=f(t,\tv,\hv)$. In the case~\eqref{kin.eq:M} $f$ solves $\partial_t f =\mu(\rho M-f)$
for which it is immediate to see by explicit integration that $f$ converges to $\rho M$ in time.
Moreover, in the case~\eqref{operator.proto}, it is possible to show under classical arguments that the operator $L$ is bounded, symmetric and non-positive, as 
\begin{equation*}\label{prop.positive}
\begin{split}
\ave{Lf,f}_M &:= \intV Lf f M^{-1} \, {\rm d} \tv {\rm d} \hv \\
&= \mathcal{D}_M(f) := -\dfrac{1}{2}\intV \intV\left(\dfrac{f(\tv,\hv)}{M(\tv,\hv)}-\dfrac{f(\tv^*,\hv^*)}{M(\tv^*,\hv^*)}\right)^2 \, {\rm d}\tv {\rm d}\tv^* {\rm d}\hv {\rm d}\hv^* \le 0.
\end{split}
\end{equation*}
The negative \textit{entropy dissipation functional} $\mathcal{D}_M(f)$ satisfies $\mathcal{D}_M(f)=0$ if $f=\rho M$ and $\mathcal{D}_M(f) \le - \mu||f-\rho M||^2_{L^2_M}$ thanks to the Jensen inequality.  This implies that
\[
\dfrac{d}{dt} ||f-\rho M||^2_{L^2_M} = \dfrac{d}{dt}||f||^2_{L^2_M}  =\ave{Lf,f}_M \le - \mu||f-\rho M||^2_{L^2_M},
\]
so that for the Gronwalls inequality ($f_0$ denoting the initial condition)
\[
||f-\rho M||_{L^2_M} \le \exp^{-\frac{\mu t}{2}}||f_0-\rho M||_{L^2_M}.
\]
The latter implies exponential entropy decay to the equilibrium $\rho M$.

 The same argument can be directly applied to direction marginal equation $\partial_t (\rho\hf)= \mathcal{L}_q(\rho \hf)$, with $\mathcal{L}_q$ defined in~\eqref{def:Lq.marg} for which we have that
\begin{equation*}
\begin{split}
\ave{\mathcal{L}_q(\rho \hf),\rho \hf}_q &:= \intS \mathcal{L}_q(\rho \hf) \rho \hf q(\hv)^{-1} \, {\rm d} \hv \\
&= \mathcal{D}_q(\rho \hf) := -\dfrac{\rho^2}{2}\intS \left(\dfrac{\hf(\hv)}{q(\hv)}-\dfrac{\hf(\hv^*)}{q(\hv^*)}\right)^2 \, {\rm d}\hv {\rm d}\hv^* \le 0.
\end{split}
\end{equation*}
Conversely, this does not hold for $\tf$ in general, as $\psi$ depends on $\hv$ and this causes the fact that the turning kernel $\psi$ does not coincide with the marginal equilibrium $\psi_q$. Only when $\psi=\psi(\tv)$, which implies that the equilibrium $\psi_q^c$ only depends on $\tv$ as~\eqref{psiv.1} holds, one can invoke the same argument and prove that
\begin{equation*}
\begin{split}
\ave{\mathcal{L}_\psi(\rho \tf),\rho \tf}_\psi &:= \intR \mathcal{L}_\psi(\rho \tf) \rho \tf \psi(\tv)^{-1} \, {\rm d} \tv \\
&= \mathcal{D}_\psi(\rho \tf) := -\dfrac{\rho^2}{2}\intR \left(\dfrac{\tf(\tv)}{\psi(\tv)}-\dfrac{\tf(\tv^*)}{\psi(\tv^*)}\right)^2 \, {\rm d}\tv {\rm d}\tv^* \le 0.
\end{split}
\end{equation*}
In the case~\eqref{psiv.1}, $T$ and $M$ coincide and ~\eqref{f.indep} holds true. However, even though $M$ and $T$ coincide, and their kernels are the same, the structure of the operators~\eqref{kin.eq:M} and~\eqref{def:Lpsi} is different. As a consequence, a different result for entropy decay is needed.
\begin{theorem}
Let us consider the system
\begin{equation*}
    \begin{cases}
        \partial_t f (t,\tv,\hv) = \mathcal{L}f\\ 
        f(0,\tv,\hv):=f_0(\tv,\hv)
    \end{cases}
\end{equation*}
where $\mathcal{L}$ is defined in~\eqref{def:L}. Let us assume that~\eqref{norm.psi},\eqref{psi.L1},\eqref{norm.q}\eqref{q.L1} hold and that $\psi$ does not depend on $\hv$, i.e.~\eqref{psiv.1} is verified. Then
\begin{enumerate}
    \item \begin{equation}\label{teo.1}
\begin{split}
\ave{\mathcal{L}f,f}_T 
=: \mathcal{D}(f)=2\mathcal{D}_T(f)-\mathcal{D}_{\psi_q}(\rho \tf)-\mathcal{D}_q(\rho \hf)  \le 0
\end{split}
\end{equation}
where
\begin{equation*}
   \mathcal{D}_T(f) := -\dfrac{1}{2}\intV \left(\dfrac{f(t,\tv,\hv)}{T(\tv,\hv)}-\dfrac{f(t,\tv^*,\hv^*)}{T(\tv^*,\hv^*)}\right)^2 \, {\rm d}\tv {\rm d}\tv^*{\rm d}\hv {\rm d}\hv^* \le 0,
\end{equation*}
which means that the operator $\mathcal{L}$ is symmetric and nonpositive.
Moreover, ~\eqref{teo.1} implies that
\[
\dfrac{d}{dt} ||f-\rho T||^2_{L^2_T} =\mathcal{D}(f) \le 0.
\]
\item If $f=\rho \tf \hf \, \forall t>0$, then we have that
\[
||f-\rho T||_{L^2_T} \le \exp^{-\frac{\hmu t}{2}}||\rho \hf_0-\rho q||_{L^2_q} + \exp^{-\frac{\tmu t}{2}}||\rho \tf_0-\rho \psi_q||_{L^2_{\psi_q}}
\]
\end{enumerate}
\end{theorem}
\begin{proof}
 Thanks to~\eqref{def:L.T}, it is immediate to prove the equality in~\eqref{teo.1}. 
 Then, thanks to~\eqref{q.L1}, one can apply the Jensen inequality to $\mathcal{D}_q(\rho\hf)$ with respect to the measure $q$ on $\mathbb{S}^{d-1}$ and show that
 \[
 -\mathcal{D}_q(\rho \hf) \le \mathcal{D}_T(f).
 \]
 As $\psi$ does not depend on $\hv$ so that~\eqref{psiv.2} holds, thanks to~\eqref{psi.L1}, we have that a similar inequality holds, i.e.
 \[
 -\mathcal{D}_{\psi_q}(\rho \tf) \le \mathcal{D}_T(f).
 \]
\end{proof}
This result clearly shows that there is entropy decay to the equilibrium~\eqref{f.indep} with~\eqref{hf.infty},\eqref{tf.infty} in the Hilbert space $L^2_T$. The entropy dissipation operator $\mathcal{D}$ clearly reflects the structure of the operator $\mathcal{L}$ highlighted in~\eqref{def:L.T}. The opposite of the marginal operators that would describe the relaxation of the marginals to their respective equilibria $q$ and $\psi_q$ contribute with a positive entropy, however they are balanced by the double negative entropy $2\mathcal{D}_T$.

In the case in which $\psi$ depends on $\hv$ we have that this result does not apply, as it is not possible to prove the first point of the Theorem. In order to have a hint of the possible entropy decay we remark that
\begin{equation}\label{entropy.cond}
\begin{split}
    || f-\rho T||^2_{L^2_T} &= \rho^2\intS ||f_c ||_{L^2_{\psi_q^c}} \dfrac{(\hf-q)^2}{q} \, {\rm d}\hv + \rho^2 \intS ||f_c-\psi_q^c||^2_{L^2_{\psi_q^c}} (2\hf(\hv)-q(\hv)) \, {\rm d} \hv\\[0.2cm]
    &\le \rho^2 \left(\sup_{\hv \in \mathbb{S}^{d-1}}||f_c||_{L^2_{\psi_q^c}}  ||\hf-q||^2_{L_q} + \sup_{\hv \in \mathbb{S}^{d-1}} ||f_c-\psi_q^c ||^2_{L^2_{\psi_q^c}} \right). 
    \end{split}
\end{equation}
where for $\phi(\tv|\hv)$
\[
||\phi||^2_{L^2_{\psi_q^c}}  := \intR \phi(\tv|\hv)^2\dfrac{1}{\psi_q^c(\tv|\hv)} \, {\rm d} \tv, \qquad \forall \hv \in \mathbb{S}^{d-1}
\]
is the norm in $L^2(\R_+)$ of the conditional density $\phi= f_c$ or of the difference $\phi=f_c-\psi_q^c$ with respect to the weight $\psi_q^c$, which is the conditional stationary equilibrium~\eqref{fc.infty}. 
In order to have exponential decay, the exponential convergence of $\hf$ is necessary (and it holds), but not sufficient. In fact,~\eqref{entropy.cond} shows that one needs the exponential entropy decay of $f_c$. 

In this general case, the convergence of the marginal $\tf$ is not easy to prove as $\psi_q\neq \psi_q^c$, so that relating the metrics of $L^2_{\psi_q}$ and $L^2_T$ is not trivial. For $f_c$, whose related metrics is more easily connected to the one defined by $T$ by virtue of~\eqref{def:T}, even though its equilibrium is known, it was not possible to prove under which conditions the exponential decay holds.


\section{Macroscopic limits}\label{sec:macro_lim}
In this section, we formally derive the macroscopic limits of the transport equation \eqref{eq:kin_strong.2}. 
We begin by considering the regime in which the two jump processes evolve on the same temporal scale, i.e., the two frequencies $\tmu,\hmu$ are of the same order. We then move to the case in which one of the two processes is faster than the other one. The resulting macroscopic equations are then compared with the classical case described by equation \eqref{kin.eq:M} under the assumption of an unconditioned probability distribution for the speed and a zero-mean probability distribution for the direction.

\subsection{Preliminaries: notation and the classical case}
The macroscopic or hydrodynamic limit of the linear transport equation~\eqref{kin.eq:M} can be performed under several scalings (diffusive, hyperbolic etc..). The common point is that hydrodynamic  limits are performed in a fluid (not rarefied) regime, which  is formally defined by a high frequency
\begin{equation}\label{hyp.1}
    \mu \rightarrow \dfrac{\mu}{\varepsilon},
\end{equation}
where $0<\varepsilon \ll 1$ is a small parameter, typically relatable to the Knudsen number. The hydrodynamic limit of Equation~\eqref{kin.eq:M}, which under the scaling~\eqref{hyp.1} is solved by $f^\varepsilon$, is then defined by $\varepsilon \to 0^+$ which corresponds to an arbitrarily large (infinite) frequency of continuous reorientations.  The hydrodynamic limit of~\eqref{kin.eq:M} under the scaling~\eqref{hyp.1}, is usually performed by considering a Chapmann-Enskog expansion of the distribution $f^\varepsilon$. In the following, we omit the explicit dependence on $(t,x)$ for notational simplicity and introduce the Chapman-Enskog expansion of $f^\varepsilon$ for each $(t,x)$ fixed, i.e.,
\begin{equation}\label{def:Ch.E.}
    f^\varepsilon(\tv,\hv) =f^0(\tv,\hv) + \varepsilon f^{\bot}(\tv,\hv),
\end{equation}
where we assume that the mass is in the leading order $f^0=\rho^0M$, while the correction carries no mass, i.e.
\begin{equation*}
    \rho^0 := \intV f^0 \, {\rm d} \hv \, {\rm d} \tv \equiv \rho, \qquad \rho^{\bot} := \intV f^{\bot} \, {\rm d} \hv \, {\rm d} \tv  \equiv 0.
\end{equation*}
The zero mass assumption on the correction $f^\bot$ is linked to the fact that the operator $L$ conserves the mass which is already at equilibrium, and it is related to the Fredholm alternative~\eqref{L2M} for $L$ given by~\eqref{operator.proto}. For this reason the mass is concentrated in $f^0$ which belongs to $\ker(L)={\rm span}(M)$. This zero mass assumption also implies that ${Lf^\bot=-\mu f^\bot}$, which allows to determine easily the correction $f^\bot$ that belongs to ${\rm span}(M)^\bot$ defined in~\eqref{spanM}.
The macroscopic limit of the transport equation~\eqref{kin.eq:M} under the scaling~\eqref{hyp.1} with first order correction is  
\begin{equation}\label{macro1_govern_dep}
    \partial_t \rho + \nabla_x \cdot \left(\rho \, u_M \right)= \dfrac{\varepsilon}{\mu} \, \nabla_x \cdot(\mathbb{D}_M\nabla_x \rho)\,.
\end{equation}
Here, the drift velocity $u_M$ is the average velocity of $M$ defined in \eqref{def:M} 
\begin{equation}\label{u_M}
    u_M:=\intV vM(\tv,\hv){\rm d}v=\intS  \hv \tilde{u}_{\psi}(\hv) q(\hv) \, {\rm d} \hv,
\end{equation}
where the quantity
\begin{equation}\label{u_psi}
    \tilde{u}_{\psi}(\hv):= \intR \tv\psi(\tv|\hv)d\tv\,
\end{equation}
represents the average speed acquired after a speed-jump by a particle moving along direction $\hv$. The average velocity $u_M$ is thus the average velocity of a particle moving on the directed structure described by $q$ according to the speed distribution $\psi$. The diffusion tensor $\mathbb{D}_M$ is defined as the variance-covariance matrix of $M$
\begin{equation*}\label{D_M}
\mathbb{D}_M:=\intV 
v\otimes(v-u_M)M(\tv,\hv)\,{\rm d}v=V_\psi \mathbb{V}_q-u_M\otimes u_M \,,
\end{equation*}
where 
\begin{equation}\label{V_psi}
    V_{\psi}:=\intR \tv^2\psi(\tv|\hv){\rm d}v
\end{equation}
is assumed to be independent on $\hv$ and
\begin{equation*}\label{V_q}
   \mathbb{V}_q:=\intS \hv\otimes\hv q(\hv)\,{\rm d}\hv\,. 
\end{equation*}
For the reader’s convenience, we also define the mean direction of the transition probability $q$ as
\begin{equation}\label{u_q}
    u_q:=\intS \hv q(\hv){\rm d}\hv\,, 
\end{equation}
which is the average direction after a reorientation. We also define the second-order variance-covariance tensor of $q$
\begin{equation}\label{mattD_q}
\mathbb{D}_q := \intS   \hv\otimes(\hv-u_q)q(\hv) \, {\rm d} \hv=\mathbb{V}_q-u_q\otimes u_q
\end{equation}
as well as the average velocity of the speed distribution $\psi_q$ 
\begin{equation}\label{u_psiq}
    u_{\psi_q}:=\intR \tv \psi_q(\tv){\rm d}\tv=\intV\tv\psi(\tv|\hv)q(\hv){\rm d}\hv{\rm d}\tv=\intS u_\psi(\hv)q(\hv){\rm d}\hv 
\end{equation}
and its second order moment
\begin{equation}\label{V_psiq}
    V_{\psi_q}:=\intR \tv^2\psi_q(\tv){\rm d}\tv=\intV\tv^2\psi(\tv|\hv)q(\hv){\rm d}\tv{\rm d}\hv\,. 
\end{equation}
Given the assumption that $V_\psi$ is independent of $\hv$, if follows that $V_\psi\equiv V_{\psi_q}$. 

Therefore, the average of the transition probability $T$ defined in~\eqref{def:T} is
\begin{equation}\label{u_T}
\begin{split}
    u_T :\,=& \intV v T(\tv,\hv) \, {\rm d} \hv \, {\rm d} \tv= \dfrac{\tmu}{\tmu+\hmu}\intS \hv\tilde{u}_{\psi}(\hv) q(\hv)  \, {\rm d} \hv+\dfrac{\hmu}{\tmu+\hmu} u_{\psi_q}u_q\\[0.3cm]
    =&\dfrac{\tmu}{\tmu+\hmu} u_M+\dfrac{\hmu}{\tmu+\hmu} u_{\psi_q}u_q.
    \end{split}
\end{equation}
The average $u_T$ shows that in the microscopic dynamics described by~\eqref{eq:kin_strong}, the average velocity has a contribution defined by the basic equilibrium average $u_M$, which is proportional to the speed frequency ratio $\frac{\tmu}{\tmu+\hmu}$, and a contribution, proportional to the reorientation frequency ratio $\frac{\hmu}{\tmu+\hmu}$, that is related to the directional average $u_q$ weighted by the average speed $u_{\psi_q}$ along the directional structure defined by $q$.
The corresponding second-order variance-covariance tensor is
\begin{equation}\label{D_T}
    \mathbb{D}_T:=\intV v\otimes \,(v-u_T)T(\tv,\hv)\,{\rm d}v=V_{\psi_q}\mathbb{V}_q-u_T\otimes u_T\,.
\end{equation}

\subsection{Jump processes with same order frequencies}
We start by considering the regime which naturally extends~\eqref{hyp.1}, that is defined by frequencies of the same order 
\begin{equation}\label{hyp.2}
    \hmu \rightarrow \dfrac{\hmu}{\varepsilon}, \qquad \tmu \rightarrow \dfrac{\tmu}{\varepsilon}.
\end{equation}
Under the scaling~\eqref{hyp.2}, the kinetic equation~\eqref{eq:kin_strong.2} becomes
\begin{equation}\label{kin_macro_hyp}
    \partial_t f^\varepsilon +v\cdot \nabla_x f^\varepsilon =\dfrac{1}{\varepsilon} \mathcal{L}f^\varepsilon.
\end{equation}
In analogy to the standard case of the previous section, we consider a Chapman-Enskog expansion~\eqref{def:Ch.E.} of $f^\varepsilon$. In the present case, as the kinetic equation is~\eqref{eq:kin_strong}, we need the corresponding Chapman-Enskog expansion for the marginals $\hf^\varepsilon,\tf^\varepsilon$.
Recalling the definition of $\tf$~\eqref{def:tf}, we integrate~\eqref{def:Ch.E.} with respect to ${\rm d} \hv$ obtaining
\begin{equation*}
    \rho \tf^\varepsilon = \rho \tf^0+\varepsilon \rho \tf^{\bot}\,.
\end{equation*}
Further integrating in ${\rm d} \tv$ we find 
\begin{equation}\label{tf.bot}
\intR \tf^{\bot} \, {\rm d} \tv =0 .
\end{equation}
Similarly, recalling the definition of $\hf$~\eqref{def:hf} and integrating~\eqref{def:Ch.E.} with respect to ${\rm d} \tv$ we obtain
\begin{equation*}
    \rho \hf^\varepsilon = \rho \hf^0+\varepsilon \rho \hf^{\bot}\,. 
\end{equation*}
Integrating this expression in ${\rm d} \hv$ we find
\begin{equation}\label{hf.bot}
\intS \hf^{\bot} \, {\rm d} \hv =0.
\end{equation}
It is worth highlighting the fact that the corrections of the marginals are both multiplied by $\rho$ and it is their mass which vanishes, i.e.,~\eqref{tf.bot}-\eqref{hf.bot}. This is the reason underlying the fact that the gain term of $\mathcal{L}$ does not vanish on $\ker(\mathcal{L})^\bot$, which makes the pseudo-inverse of $\mathcal{L}$ less standard~\eqref{pseudo.inv.L}.

From equation \eqref{kin_macro_hyp}, collecting the terms of order $\varepsilon^0$ we have that
\[
\mathcal{L}f^0 = 0,
\]
which implies $f^0 \in \ker(\mathcal{L})$. According to the results derived in the previous section, this yields
\begin{equation}\label{f0.1}
    f^0 = \rho T
\end{equation}
where $T$ is defined in~\eqref{def:T}. At order $\varepsilon^1$ equation \eqref{kin_macro_hyp} becomes
\begin{equation}\label{macro1.eps1}
    \partial_t f^0 + v\cdot \nabla_x f^0 =\mathcal{L}f^{\bot}.
\end{equation}
Integrating the latter over $\mathcal{V}$ with $f^0$ defined in~\eqref{f0.1} we obtain the macroscopic equation
\begin{equation}\label{macro1_govern}
    \partial_t \rho +\nabla_x \cdot \left(\rho u_T\right)=0,
\end{equation}
where $u_T$ is defined in~\eqref{u_T}. 

We now look for the correction $f^{\bot}$, which belongs to ${\rm span} \lbrace T\rbrace^{\perp}$ and has zero mass according to the scalar product~\eqref{sc.pr.T}. From~\eqref{macro1.eps1} and using the pseudo inverse operator $\mathcal{L}^{-1}$ defined in~\eqref{pseudo.inv.L} we have that
\begin{equation}\label{f_perp.1}
    f^{\bot} = - \dfrac{1}{\tmu+\hmu} \left(\partial_t f^0 + v\cdot \nabla_x f^0\right) +\dfrac{\hmu}{\tmu+\hmu} \rho \tf^{\bot}q+\dfrac{\tmu}{\tmu+\hmu} \rho \hf^{\bot}\psi\,. 
\end{equation}
We first determine $\rho \hf^{\bot}$ by inverting $\hat{\mathcal{L}}_q$ according to~\eqref{Lq.inv} in~\eqref{eq:hf.bis}, which gives
\begin{equation*}
 \rho \hf^{\bot}=-\dfrac{1}{\hmu} \left(\partial_t \rho \hf^0 + \hv \cdot \nabla_x \intR \tv f^0(\tv,\hv) \, {\rm d} \tv \right)\,, 
\end{equation*} 
and can be written as
\begin{equation}\label{hf_perp}
 \rho \hf^{\bot}=\dfrac{1}{\hmu}q(\hv) \nabla_x \cdot \Big[\left(\rho u_T\right)  - \dfrac{1}{\hmu+\tmu}\left(\tmu u_\psi(\hv)+\hmu u_{\psi_q}\right) \hv\rho \Big],     
\end{equation}
where $u_\psi(\hv)$ and $u_{\psi_q}$ are defined in \eqref{u_psi} and \eqref{u_psiq}, respectively. Similarly, from~\eqref{eq:tf.bis}, inverting $\tilde{\mathcal{L}}_{\psi}$ according to~\eqref{Lpsi.inv} gives
\begin{equation*}
    \rho \tf^{\bot} = -\dfrac{1}{\tmu} \left[\partial_t \rho \tf^0 + \tv \nabla_x \cdot \intS \hv f^0(\tv,\hv) \, {\rm d} \hv\right] +\rho\intS\psi(\tv|\hv)\hf^{\bot}(\hv){\rm d}\hv\,.
\end{equation*}
Substituting the expression from~\eqref{hf_perp} in the latter and rearranging the terms we obtain
\begin{equation}\label{tf_perp}
\begin{split}
    \rho \tf^{\bot} &= \dfrac{1}{\tmu}\psi_q(\tv) \nabla_x \cdot (\rho u_T) -\dfrac{1}{\hmu+\tmu} \nabla_x \cdot\left[\rho\intS \tv\psi(\tv|\hv)\hv q(\hv){\rm d}\hv\right]\\[0.3cm]
    &\,\,\,\,+\dfrac{1}{\hmu}\psi_q(\tv) \nabla_x \cdot (\rho u_T) -\dfrac{\tmu}{\hmu+\tmu}\nabla_x\cdot\left[\rho\intS\hv q(\hv)\psi(\tv|\hv)\tilde{u}_{\psi}(\hv){\rm d}\hv\right]\\[0.3cm]
    &\,\,\,\,-\dfrac{1}{\hmu+\tmu}\nabla_x\left[\rho u_{\psi_q}\intS\hv q(\hv)\psi(\tv|\hv){\rm d}\hv\right] - \dfrac{\hmu}{\hmu+\tmu}\nabla_x\cdot \left[\rho u_q\tv\psi_q(\tv)\right]
    \end{split}
\end{equation}
where $u_q$ is defined in~\eqref{u_q}. By plugging~\eqref{tf_perp} and~\eqref{hf_perp} into~\eqref{f_perp.1}, we obtain the expression for $f^{\bot}$, namely
\begin{equation}\label{f_perp.2}
\begin{split}
    f^{\bot} =& -\dfrac{1}{\tmu+\hmu} \left[v\cdot \nabla_x (\rho T)-T\,\nabla_x \cdot (\rho u_T) \right] +\dfrac{1}{\hmu}T\,\nabla_x\cdot(\rho u_T)\\[0.3cm]
    &+\dfrac{\hmu}{\tmu+\hmu}\psi_q(\tv)q(\hv)\nabla_x\cdot(\rho u_T)-\dfrac{\tmu}{\hmu(\tmu+\hmu)^2}\psi(\tv|\hv)q(\hv)\nabla_x\cdot\left[\rho\hv(\tmu u_\psi(\hv)+\hmu u_{\psi_q})\right] \\[0.3cm]
    &-\dfrac{\hmu^2}{\tmu(\tmu+\hmu)^2}q(\hv)\nabla_x\cdot\left[\rho u_q\tv\psi_q(\tv)\right]-\dfrac{\hmu}{(\hmu+\tmu)^2}q(\hv)\nabla_x\cdot\left[\rho\intS\tv\psi(\tv|\hv)\hv q(\hv){\rm d}\hv\right]\\[0.3cm]
    &-\dfrac{\tmu}{(\tmu+\hmu)^2}q(\hv)\nabla_x\cdot\left[\rho\intS\hv q(\hv)\psi(\tv|\hv)\tilde{u}_{\psi}(\hv){\rm d}\hv\right]\\[0.3cm]
    &-\dfrac{\hmu}{(\tmu+\hmu)^2}q(\hv)\nabla_x\cdot\left[\rho u_{\psi_q}\intS\hv q(\hv)\psi(\tv|\hv){\rm d}\hv\right]\,
    \end{split}
\end{equation}
which satisfies 

$$
\intV f^{\bot} \, {\rm d} v  \equiv 0.
$$
To derive the correction term in the macroscopic equation \eqref{macro1_govern}, we need to integrate
\[
\partial_t f^0 +v\cdot \nabla_x (f^0+\varepsilon f^{\bot}) =\mathcal{L}f^{\bot}
\]
over $\mathcal{V}$ that gives (computations are reported in the Appendix \ref{appendix2})
\begin{equation}\label{macro1_govern_corr}
\begin{split}
    \partial_t \rho + \nabla_x \cdot (\rho \, u_T)& =  \dfrac{\varepsilon}{\tmu+\hmu} \, \nabla_x \cdot(V_{\psi_q}\mathbb{V}_q \nabla_x\rho)-\varepsilon\dfrac{2\hmu+\tmu}{\hmu(\tmu+\hmu)}\nabla_x\cdot[u_T\otimes u_T\nabla_x\rho] \\[0.3cm]
    &+\varepsilon\dfrac{\hmu^2}{\tmu(\hmu+\tmu)^2}\nabla_x\cdot[(V_{\psi_q}-u_{\psi_q}^2) u_q\otimes u_q\nabla_x\rho]\\[0.3cm]
    &+\dfrac{\varepsilon}{(\hmu+\tmu)^2}\nabla_x\cdot\left[u_q\otimes\intS (\tmu u_\psi^2(\hv)+\hmu V_{\psi_q})\,\hv\, q(\hv){\rm d}\hv\,\nabla_x\rho\right]\\[0.3cm]
        &+\varepsilon\dfrac{\tmu}{\hmu(\hmu+\tmu)^2}\nabla_x\cdot\left[\intS \left(\hmu u_{\psi_q} u_\psi(\hv)+\tmu u_\psi^2(\hv)\right) q(\hv)\hv\otimes\hv\,{\rm d}\hv\,\nabla_x\rho\right]\,.\\[0.3cm]
    \end{split}
\end{equation}
Of course, the diffusive correction term of order $\varepsilon$ includes mixed terms taking into account the superposition of the microscopic mechanisms for the speed and the direction, which are not uncorrelated due to the mechanism embodied by $\psi(\tv|\hv)$.
In order to get an insight on this correction term, we consider the simplified case in which $\psi(\tv|\hv)=\psi(\tv)$. The straightforward consequence~\eqref{psiv.1} implies that ${u_\psi(\hv)=u_{\psi_q}}$, and that the variance of $\psi=\psi_q$ is $(V_{\psi_q} -u_{\psi_q}^2)$, while $u_T=u_{\psi_q}u_q$, and $\mathbb{D}_T=V_{\psi_q}\mathbb{V}_q-u_{\psi_q}^2u_q\otimes u_q$. Then, the corresponding macroscopic equation with first order correction reads
\begin{equation}\label{macro1_corr_simpl}
\begin{aligned}[b]
    \partial_t \rho + \nabla_x \cdot (\rho \, u_T)& =  \dfrac{\varepsilon}{\tmu+\hmu} \, \nabla_x \cdot\left[\left(V_{\psi_q}\mathbb{V}_q-u_{\psi_q}^2u_q\otimes u_q\right) \nabla_x\rho\right]\\[0.3cm]
    &+\varepsilon\dfrac{\hmu}{\tmu(\hmu+\tmu)}\nabla_x\cdot\left[V_{\psi_q} u_q\otimes u_q\nabla_x\rho\right]+\varepsilon\dfrac{\tmu}{\hmu(\hmu+\tmu)}\nabla_x\cdot[u_{\psi_q}^2 \mathbb{V}_q\nabla_x\rho]\\[0.3cm]
    &-\varepsilon\dfrac{\hmu^2+\tmu^2}{\hmu\tmu(\hmu+\tmu)}\nabla_x\cdot\left[u_{\psi_q}^2u_q\otimes u_q\nabla_x\rho\right]\,.
    \end{aligned}
\end{equation} 
Choosing $\tmu=\hmu=1$, we have that
\begin{equation*}
\begin{aligned}[b]
    \partial_t \rho + \nabla_x \cdot (\rho \, u_T)& =  \dfrac{\varepsilon}{2} \, \nabla_x \cdot\left[\mathbb{D}_T \nabla_x\rho\right]\\[0.3cm]
    &+\dfrac{\varepsilon}{2}\nabla_x\cdot\left[(V_{\psi_q} -u_{\psi_q}^2)u_q\otimes u_q\nabla_x\rho\right]+\dfrac{\varepsilon}{2}\nabla_x\cdot[u_{\psi_q}^2 \mathbb{D}_q\nabla_x\rho].
    \end{aligned}
\end{equation*} 
In the latter, the first term in the right-hand side is the diffusive correction related to the equilibrium $T$, while the second two terms are related to the diffusion (variance-covariance) of the two independent jump processes: the speed one ($V_{\psi_q} -u_{\psi_q}^2$) directed by the tensor $u_q\otimes u_q$, and the direction one ($\mathbb{D}_q$) with quadratic speed $u_{\psi_q}^2$ along the directed structure $q$. 

This helps interpreting the correction term in~\eqref{macro1_corr_simpl}, which can be rewritten as
\begin{equation*}
\begin{aligned}[b]
    \partial_t \rho + \nabla_x \cdot (\rho \, u_T)& =  \dfrac{\varepsilon}{\tmu+\hmu} \, \nabla_x \cdot\left[\left(V_{\psi_q}\mathbb{V}_q-u_{\psi_q}^2u_q\otimes u_q\right) \nabla_x\rho\right]\\[0.3cm]
    &+\varepsilon\dfrac{\hmu}{\tmu(\hmu+\tmu)}\nabla_x\cdot\left[V_{\psi_q} u_q\otimes u_q\nabla_x\rho\right]+\varepsilon\dfrac{\tmu}{\hmu(\hmu+\tmu)}\nabla_x\cdot[u_{\psi_q}^2 \mathbb{V}_q\nabla_x\rho]\\[0.3cm]
    &-\varepsilon\dfrac{\hmu}{\tmu(\hmu+\tmu)}\nabla_x\cdot\left[u_{\psi_q}^2u_q\otimes u_q\nabla_x\rho\right]-\varepsilon\dfrac{\tmu}{\hmu(\hmu+\tmu)}\nabla_x\cdot\left[u_{\psi_q}^2u_q\otimes u_q\nabla_x\rho\right]\,.
    \end{aligned}
\end{equation*} 
In the latter, the first term in the right-hand side is again the diffusive correction related to the equilibrium $T$, the second and fourth terms are related to the diffusion in the speed jump process, while the third and fifth terms are connected to the diffusion in the direction jump process. 

\subsection{Jump operators with different order frequencies}
We now consider the regime in which one of the two operators is faster than the other one, i.e., the two microscopic processes happen on two different time scales. Specifically, we analyze the following scalings:
\begin{equation}\label{hyp.3}
\begin{split}
   (a)\qquad \hmu \rightarrow \dfrac{\hmu}{\varepsilon^2}, \qquad \tmu \rightarrow \dfrac{\tmu}{\varepsilon}\,,\\[0.3cm]
   (b)\qquad \hmu \rightarrow \dfrac{\hmu}{\varepsilon}, \qquad \tmu \rightarrow \dfrac{\tmu}{\varepsilon^2},
   \end{split}
\end{equation}
where $(a)$ represents a regime scenario where reorientations happen more frequently, while $(b)$ accounts for more frequent changes in speed. Because of the different order of the frequencies, here we adopt a Hilbert expansion for $f^\varepsilon$ which takes into account higher order in $\varepsilon$ corrections. It is defined for fixed $(t,x)$ by
\begin{equation}\label{def:H.}
    f^\varepsilon(\tv,\hv)= f^0(\tv,\hv)+\varepsilon f^1(\tv,\hv) +\varepsilon^2 f^2(\tv,\hv)+\mathcal{O}(\varepsilon^2),
\end{equation}
with
\begin{equation*}
    \rho^0 := \intV f^0 \, {\rm d} \tv \, {\rm d} \hv \equiv \rho, \quad \rho^i := \intV f^i \, {\rm d} \tv \, {\rm d} \hv \equiv 0 \quad \forall i \ge 1,
\end{equation*}
i.e., the mass is in the leading order $f^0$, while higher order corrections carry no mass.
Again, as~\eqref{eq:kin_strong} features the marginals, we need to define a Hilbert expansion for $\tf^\varepsilon$ and $\hf^\varepsilon$ which will follow from~\eqref{def:H.}. Integrating~\eqref{def:H.} with respect to ${\rm d} \hv$ over $\mathbb{S}^{d-1}$ we obtain
\begin{equation*}
    \rho \tf^\varepsilon = \rho \tf^0+\varepsilon \rho \tf^1 +\varepsilon^2 \rho \tf^2 +\mathcal{O}(\varepsilon^2)
\end{equation*}
and, from this, integrating in ${\rm d} \tv$ over $\R_+$, we deduce that 
\[
\intR \tf^i \, {\rm d} \tv =0 \quad \forall i \ge 1.
\]
Similarly, integrating~\eqref{def:H.} with respect to ${\rm d} \tv$ over $\R_+$ we obtain
\begin{equation*}
    \rho \hf^\varepsilon = \rho \hf^0+\varepsilon \rho \hf^1 +\varepsilon^2 \rho \hf^2 +\mathcal{O}(\varepsilon^2)
\end{equation*}
so that, integrating the latter in ${\rm d} \hv$ over $\mathbb{S}^{d-1}$, we find
\[
\intS \hf^i \, {\rm d} \hv =0 \quad \forall i \ge 1\,.
\]
Starting from case~\eqref{hyp.3}$(a)$, the kinetic equation becomes
\begin{equation}\label{kin.eq.macro2}
  \partial_t f^\varepsilon +v \cdot \nabla_x f^\varepsilon = \dfrac{\hmu}{\varepsilon^2}\hat{\mathcal{L}}f^\varepsilon +  \dfrac{\tmu}{\varepsilon}\tilde{\mathcal{L}} f^\varepsilon\,.
  \end{equation}
At order $\varepsilon^0$, we have
\begin{equation*}
\hat{\mathcal{L}}f^0=0,    
\end{equation*}
and, therefore, we deduce that
\begin{equation}\label{macro2.f0}
    f^0(\tv,\hv)=\rho \tf^0(\tv) q(\hv)\,.
\end{equation}
Integrating $\hat{\mathcal{L}}$ with respect to ${\rm d} \hv$ reveals that the leading-order operator conserves not only the total mass $\rho$, but also the quantity $\rho \tf(\tv)$ for all $\tv$.
Then, integrating \eqref{macro2.f0} with respect to ${\rm d}\tv$ yields
\begin{equation}\label{macro2.hf0}
     \hf^0(\hv)= q(\hv).
\end{equation}
At order $\varepsilon^1$, equation~\eqref{kin.eq.macro2} becomes
\begin{equation*}
\hmu\hat{\mathcal{L}}f^1 +  \tmu\tilde{\mathcal{L}} f^0=0,  
\end{equation*}
which can be rewritten as
\begin{equation}\label{macro2.eps1}
\hmu(\rho \tf^1 q-f^1)+\tmu (\rho \hf^0 \psi-f^0)=0.  \end{equation}
Integrating~\eqref{macro2.eps1} over ${\rm d}\hv$ and using~\eqref{macro2.hf0}, we obtain
\begin{equation}\label{macro2.tf0}
    \tf^0(\tv)=\psi_q(\tv).
\end{equation}
Combining~\eqref{macro2.f0} and~\eqref{macro2.tf0}, we obtain the explicit form for $f^0$, namely
\begin{equation}\label{macro2.f0_bis}
   f^0(\tv,\hv)=\rho \psi_q(\tv) q(\hv)\, .
\end{equation}
Since $f^0$ is a joint distribution for the pair $(\tv, \hv)$, in view of~\eqref{f:cond} which can be written for $f^0$, equation~\eqref{macro2.f0_bis} implies
\[
f_c^0(\tv|\hv)=\tf^0(\tv)=\psi_q(\tv)\,.
\]
This is coherent with the fact that the equilibrium in $\hv$ is reached more fastly in regime (a). Moreover,
it indicates that there is statistical independence at leading order, due to the separation of time scales. This does not imply unconditional independence as the transition probability $\psi(\tv|\hv)$ always encodes dependence on $\hv$ along the whole dynamics. However, in this regime, the system relaxes 
 fast toward an equilibrium distribution which features  statistical independence of the speed and direction.
 
Integrating~\eqref{macro2.eps1} with respect to ${\rm d}\tv$ we obtain
\begin{equation}\label{macro2.hf1}
    \hf^1=0,
\end{equation}
which means that the marginal $\hf$ is already at equilibrium, which is consistent with the dominance of the faster reorientation process. Plugging~\eqref{macro2.f0_bis} in~\eqref{macro2.eps1} we obtain
\begin{equation*}
    f^1=\rho \tf^1 q +\dfrac{\tmu}{\hmu}\rho q\left(\psi-\psi_q\right)\,.
\end{equation*}
At order $\varepsilon^2$, equation \eqref{kin.eq.macro2} becomes
\begin{equation}\label{macro2.eps2}
    \partial_t f^0 +v\cdot \nabla_x f^0 = \hmu \hat{\mathcal{L}}f^2+\tmu \tilde{\mathcal{L}}f^1\,.
\end{equation}
Integrating over $\mathcal{V}$ we derive the macroscopic equation
\begin{equation}\label{macro2.govern}
    \partial_t \rho +\nabla_x \cdot (\rho u_q u_{\psi_q})=0\,,
\end{equation}
where $u_q$ and $u_{\psi_q}$ are defined in~\eqref{u_q} and~\eqref{u_psiq}. To determine the missing information about the marginal $\tf^1$ and, consequently, $f^1$, we integrate~\eqref{macro2.eps2} with $f^0$ given by~\eqref{macro2.f0_bis} over ${\rm d}\hv$:
\[
\psi_q(\tv)\partial_t \rho  +\nabla_x \cdot \left(\rho u_q\tv\psi_q(\tv)  \right) = \tmu \left(\rho \intS \psi(\tv|\hv) \hat{f}^1(\hv) \, {\rm d} \hv - \rho \tf^1(\tv)\right)\,.
\]
Using~\eqref{macro2.hf1}, the latter simplifies to
\[
\partial_t (\rho \psi_q) +\nabla_x \cdot (\rho u_q\tv\psi_q(\tv)  ) = \tmu ( - \rho \tf^1)\,
\]
and, then, including~\eqref{macro2.govern} we obtain
\begin{equation*}
    \rho\tf^1(\tv)= -\dfrac{1}{\tmu}\nabla_x \cdot (\rho u_q\tv\psi_q(\tv)  )+\dfrac{1}{\tmu}\psi_q(\tv)\nabla_x \cdot (\rho u_q u_{\psi_q})\,.
\end{equation*}
Hence,
\begin{equation*}
    f^1(\tv,\hv)= -\dfrac{1}{\tmu}q(\hv)\nabla_x \cdot \left[\rho u_q\psi_q(\tv) (\tv-u_{\psi_q})\right]+\dfrac{\tmu}{\hmu}\rho q(\hv)(\psi(\tv|\hv)-\psi_q(\tv)).
\end{equation*}
Therefore, the first order correction to~\eqref{macro2.govern} is obtained by integrating
\[
\partial_t f^0 +v\cdot \nabla_x (f^0+\varepsilon f^1) = \hmu \hat{\mathcal{L}}f^2+\tmu \tilde{\mathcal{L}}f^1
\]
over $\mathcal{V}$. This yields the macroscopic equation with first-order correction
\begin{equation}\label{macro2.govern_corr_baru}
    \partial_t \rho +\nabla_x \cdot \left[\rho\left(u_q u_{\psi_q}\left(1-\varepsilon\dfrac{\tmu}{\hmu}\right)+\varepsilon\dfrac{\tmu}{\hmu}u_M\right)\right]=\dfrac{\varepsilon}{\tmu}\nabla_x \cdot \left[(V_{\psi_q}-u_{\psi_q}^2)u_q\otimes u_q\nabla_x \rho\right],
\end{equation}
where $u_M$ is given by \eqref{u_M}. In~\eqref{macro2.govern_corr_baru} the first order correction enters in the drift term. While the leading order term features the factorized advective velocity $u_q u_{\psi_q}$, its correction is proportional to the difference $u_M-u_q u_{\psi_q}$. The diffusive correction is related to the variance of $\psi_q$ along the direction imposed by the tensor $u_q\otimes u_q$, while no diffusion related to the direction-jump process appears because under scaling (a), the equilibrium in $\hv$ is reached fast. 

Conversely, considering the scaling~\eqref{hyp.3}$(b)$, the kinetic equation becomes
\begin{equation}\label{kin.eq.macro3}
  \partial_t f^\varepsilon +v \cdot \nabla_x f^\varepsilon = \dfrac{\hmu}{\varepsilon}\hat{\mathcal{L}}f^\varepsilon +  \dfrac{\tmu}{\varepsilon^2}\tilde{\mathcal{L}} f^\varepsilon\,.
  \end{equation}
In this regime, the leading order operator $\tilde{\mathcal{L}}$ conserves the marginal $\rho \hf (\hv)$ for all $\hv$. Thus, at order $\varepsilon^0$ we have
\begin{equation*}
\tilde{\mathcal{L}}f^0=0    
\end{equation*}
implying that
\begin{equation*}
    f^0(\tv,\hv)=\rho \hf^0(\hv) \psi(\tv|\hv)\,.
\end{equation*}
At order $\varepsilon^1$, equation \eqref{kin.eq.macro3} reads
\begin{equation*}
\hmu\hat{\mathcal{L}}f^0 +  \tmu\tilde{\mathcal{L}}f^1=0    
\end{equation*}
i.e.,
\begin{equation}\label{macro3.eps1}
\hmu(\rho \tf^0 q-f^0)+\tmu (\rho \hf^1 \psi-f^1)=0.    
\end{equation}
Integrating the latter with respect to ${\rm d} \tv$ we obtain
\begin{equation*}\label{macro3.hf0}
    \hf^0(\hv)=q(\hv)
\end{equation*}
and, thus, 
\begin{equation}\label{macro3.f0}
     f^0(\tv,\hv)=\rho q(\hv) \psi(\tv|\hv)\,.
\end{equation}
Integrating~\eqref{macro2.eps1} in ${\rm d} \hv$ we obtain
\begin{equation}\label{macro3.tf0}
  \tf^0(\tv)=\intS \hf^0(\hv)\psi(\tv|\hv){\rm d}\hv=\psi_q(\tv)\,.
    \end{equation}
Unlike the previous scaling $(a)$, here the conditional distribution $f_c^0(\tv|\hv)$ at leading order is
\[
f_c^0(\tv|\hv)=\psi(\tv|\hv)
\]
which implies that 
$$
f^0(\tv,\hv)\ne\hf^0(\hv)\tf^0(\tv)\,,
$$
which means no statistical independence for $f^0$.
Thus, this scaling preserves the statistical dependence in the speed-jump operator at leading order. Moreover, plugging~\eqref{macro3.f0} and~\eqref{macro3.tf0} into \eqref{macro3.eps1} gives
\begin{equation*}
    f^1(\tv,\hv)=\rho \hf^1(\tv) \psi(\tv|\hv) +\dfrac{\hmu}{\tmu}\rho q(\hv) \left(\psi_q(\tv)-\psi(\tv|\hv)\right)\,.
\end{equation*}
At order $\varepsilon^2$, equation \eqref{kin.eq.macro3} becomes
\begin{equation}\label{macro3.eps2}
    \partial_t f^0 +v\cdot \nabla_x f^0 = \hmu \hat{\mathcal{L}}f^1+\tmu \tilde{\mathcal{L}}f^2\,.
\end{equation}
Integrating over $\mathcal{V}$ we obtain the macroscopic equation
\begin{equation*}\label{macro3.govern}
    \partial_t \rho +\nabla_x \cdot (\rho u_M)=0,
\end{equation*}
with $u_M$ given by~\eqref{u_M}. To exploit \eqref{macro3.eps2} for determining $\hf^1$, we integrate it with respect to ${\rm d} \tv$ with $f^0$ given by~\eqref{macro2.f0_bis}:
\[
q(\tv)\partial_t \rho  +\nabla_x \cdot \left(\rho q(\hv)\hv u_\psi(\hv)  \right) = \hmu \left(\rho q(\hv) - \rho \hf^1\right),
\]
i.e., 
\[
\rho \hf^1(\hv)=\rho\hf^0(\hv)+\dfrac{1}{\hmu}q(\hv)\left[\nabla_x \cdot(\rho u_M) -\nabla_x\cdot (\hv u_\psi(\hv)\rho)\right]\,.
\]
Thus, $f^1$ takes the form
\begin{equation*}
    f^1(\tv,\hv)= \rho q(\hv)\psi(\tv|\hv)+\dfrac{1}{\hmu}q(\hv)\psi(\tv|\hv)\nabla_x\cdot \left[\rho(u_M-\hv u_\psi(\hv))\right] +\dfrac{\hmu}{\tmu}\rho q(\hv) \left(\psi_q(\tv)-\psi(\tv|\hv)\right)  \,.  
\end{equation*}
Therefore, the macroscopic equation with first-order correction reads
\begin{equation}\label{macro3.govern_corr}
\begin{split}
    &\partial_t \rho +\nabla_x \cdot \left[\rho \left(u_M(1+\varepsilon)-\dfrac{\hmu}{\tmu}\varepsilon(u_M-u_qu_{\psi_q})\right)\right]=\\[0.3cm]
    &\dfrac{\varepsilon}{\hmu}\nabla_x \cdot \left[\left(\intS u_\psi^2(\hv)q(\hv)\hv\otimes\hv{\rm d}\hv- u_M\otimes u_M\right) \nabla_x\rho\right],
    \end{split}
\end{equation}
where the dominant advective velocity is now $u_M$ defined in~\eqref{u_M} which mirrors the statistical dependence of the equilibrium~\eqref{macro3.f0} under scaling (b). Here again  the first order correction also enters the flux with $u_{\psi_q}u_q$ (defined by~\eqref{u_psiq} and~\eqref{u_q}), which means a higher order independence of the two jump processes. The diffusion correction in the right-hand side does not feature any specific variance, but fluctuations related to both  direction and speed processes are present, as the speed process is faster under scaling (b), but the mechanism $\psi(\tv|\hv)$ features a dependence on $\hv$ which evolves more slowly. 
\begin{remark}
We observe that the two scalings, $(a)$ and $(b)$, correspond to distinct equilibrium distributions. Specifically, for the scaling $(a)$ the equilibrium distribution is given by 
\[
P(\tv,\hv)= \psi_q(\tv) q(\hv)\,,
\]
which features a statistical independence of direction and speed, while, for the scaling $(b)$ the equilibrium distribution is given by $M$ defined in \eqref{def:M}. These allow us to define the corresponding average velocities:
\[
u_{P}:= \intV vP(\tv,\hv) \, {\rm d} \tv \, {\rm d} \hv = u_q u_{\psi_q}\,\qquad u_{M}:= \intV vM(\tv,\hv) \, {\rm d} \tv \, {\rm d} \hv \,
\]
and second-order tensors
\[
\mathbb{D}_{P} := \intV  v\otimes (v-u_{P})P(\tv,\hv)\,  {\rm d} \tv \, {\rm d} \hv = V_{\psi_q} \mathbb{V}_q - u_{P} \otimes u_{P} \,,
\]
and
\[
\mathbb{D}_{M} := \intV  v\otimes (v-u_{M})  M(\tv,\hv)\,  {\rm d} \tv \, {\rm d} \hv = V_{\psi_q} \mathbb{V}_q - u_{M} \otimes u_{M} \,.
\]
Because of the independence of $\tv$ and $\hv$ in $P$, $u_P$ is factorized and 
\[
u_{P}\otimes u_{P} =u_{\psi_q}^2 u_q\otimes u_q\,.
\]
Then equation~\eqref{macro2.govern_corr_baru} corresponding to the faster reorientation can be rewritten as
\begin{equation*}\label{macro2.govern_corr_baru_bis}
    \partial_t \rho +\nabla_x \cdot \Big[\rho\Big( u_{P}  +\varepsilon \dfrac{\tmu}{\hmu} \Big(u_M-u_P\Big)\Big]=\dfrac{\varepsilon}{\tmu}\nabla_x \cdot \nabla_x \cdot \left(\rho \mathbb{D}_{P}-\rho V_{\psi_q}\mathbb{D}_q\right)\,,
\end{equation*}
where the diffusion is given by the equilibrium diffusion $\mathbb{D}_{P}$ minus the direction diffusion $\mathbb{D}_q$ weighted by the speed second moment $V_{\psi_q}$.
Conversely, equation~\eqref{macro3.govern_corr}  corresponding to the faster speed jump process takes the form
\begin{equation*}\label{macro3.govern_corr_bis}
    \begin{split}
    &\partial_t \rho +\nabla_x \cdot \left[\rho \left(u_{M}+\varepsilon\Big(\dfrac{\hmu}{\tmu}(u_{P}-u_M)+u_M\Big)\right)\right]=\\[0.3cm]
    &\dfrac{\varepsilon}{\hmu}\nabla_x \cdot \nabla_x \cdot \left(\rho \mathbb{D}_{M}-\rho \intS \left(V_{\psi_q}-u_\psi^2(\hv)\right)q(\hv)\hv\otimes\hv{\rm d}\hv\right)\,,
    \end{split}
\end{equation*}
where the diffusion is given by the equilibrium one $\mathbb{D}_{M}$ minus the speed variance along the direction field $q$, i.e. the term $\intS \left(V_{\psi_q}-u_\psi^2(\hv)\right)q(\hv)\hv\otimes\hv{\rm d}\hv$.
\end{remark}
\subsection{Comparison of the macroscopic limiting equations}
We now focus on the special case where the directional distribution has zero mean, i.e., $u_q=0$, and compare the macroscopic equations with first-order corrections derived for the dependent operator case~\eqref{macro1_govern_dep}, as well as for the case of independent operators for speed and direction under the same-order frequency scaling~\eqref{macro1_govern_corr}, and for the two scalings with different frequencies, namely equations~\eqref{macro2.govern_corr_baru} and~\eqref{macro3.govern_corr}. 

The case $u_q=0$ is particularly relevant due to its applicability in modeling cell migration on ECM of non-polarized fibers (both senses on a given direction can be considered). In this setting, the condition $u_q=0$ implies that the tensor $\mathbb{D}_q$ introduced in~\eqref{mattD_q} simplifies to $\mathbb{D}_q=\mathbb{V}_q$. Likewise, the effective transport velocity becomes $u_T=\tmu/(\tmu+\hmu)u_M$, i.e., it is directly proportional to the drift velocity $u_M$ appearing in \eqref{macro1_govern_dep}. 
Under these conditions, the macroscopic equation~\eqref{macro1_govern_dep} corresponding to the standard case of dependent operators remains unchanged. When  $u_q=0$ the macroscopic equation~\eqref{macro1_govern_corr} for the case of independent speed and direction operators simplifies to
\begin{equation}\label{macro_Indp_SameFreq_Eq0}
\begin{split}
    \partial_t \rho +\dfrac{\tmu}{\tmu+\hmu} \nabla_x \cdot (\rho \, u_M) &=  \dfrac{\varepsilon}{\tmu+\hmu} \, \nabla_x \cdot\left(\left(V_{\psi_q}\mathbb{V}_q-\dfrac{\tmu^2(2\hmu+\tmu)}{\hmu(\tmu+\hmu)^2}u_M\otimes u_M\right) \nabla_x\rho\right)\\[0.3cm]
    &+\varepsilon\dfrac{\tmu}{\hmu(\hmu+\tmu)^2}\nabla_x\cdot \left[\intS \left(\hmu u_{\psi_q} u_\psi(\hv)+\tmu u_\psi^2(\hv)\right) q(\hv)\hv\otimes\hv\,{\rm d}\hv\,\nabla_x\rho\right]\,.
    \end{split}
\end{equation}
Moving to the case of independent operators with frequencies of different order, the resulting macroscopic behavior differs significantly. For the scaling $(a)$ (where speed changes are slower), the macroscopic equation \eqref{macro2.govern_corr_baru} does not retain any advective terms at leading order, while its first-order correction simplifies to a transport-dominated form:
\begin{equation*}\label{macro2.Eq0}
    \partial_t \rho +\varepsilon\dfrac{\tmu}{\hmu}\nabla_x \cdot \left(\rho\, u_M\right)=0
\end{equation*}
This behavior is consistent with the assumption that reorientation dominates the dynamics, and, in the case of a zero-mean direction distribution, suppresses drift at leading order. The residual transport at first order thus reflects the delayed relaxation of the speed distribution.
For scaling $(b)$ (where speed changes are faster), the governing equation~\eqref{macro3.govern_corr} becomes
\begin{equation}\label{macro3.Eq0}
\begin{split}
    &\partial_t \rho +\nabla_x \cdot \left[\rho \,u_M\left(1+\varepsilon\dfrac{\tmu-\hmu}{\tmu}\right)\right]=\dfrac{\varepsilon}{\hmu}\nabla_x \cdot \left[\left(\intS u_\psi^2(\hv)q(\hv)\hv\otimes\hv{\rm d}\hv- u_M\otimes u_M\right) \nabla_x\rho\right]
    \end{split}
\end{equation}
which retains both advective and diffusive contributions. Notably, the fast relaxation of the speed variable allows the drift velocity $u_M$ to persist at leading order, similarly to the standard case~\eqref{macro1_govern_dep} and the same-order frequency scenario~\eqref{macro_Indp_SameFreq_Eq0}.  The diffusive correction also aligns with that of~\eqref{macro1_govern_dep}, but with a key distinction: the second-order moment of $\psi$, i.e., $V_\psi$  used in~\eqref{macro1_govern_dep} is here replaced by the squared mean $u_\psi^2$. This substitution reflects the fact that, under fast speed dynamics, fluctuations in the speed distribution are effectively averaged out, leading to a diffusion term governed by the mean-square speed rather than its variance.

If we further assume that the probability distribution $\psi(\tv|\hv)$ is independent of the direction $\hv$, then the effective drift velocity $u_M$ in equation~\eqref{macro1_govern_dep} vanishes. In this case, the macroscopic equation simplifies to a purely diffusion model:
\begin{equation*}\label{macro_Dep_uncond_Eq0}
    \partial_t \rho = \dfrac{\varepsilon}{\mu} \, \nabla_x\cdot\nabla_x\cdot (V_{\psi_q}\mathbb{V}_q \,\rho)\,,
\end{equation*}
describing diffusion modulated by the microscopic speed variance $V_{\psi_q}$ and directional variance-covariance tensor $\mathbb{V}_q$. In this same setting, the macroscopic equation \eqref{macro_Indp_SameFreq_Eq0} for independent speed and direction dynamics simplifies to:
\begin{equation*}\label{macro_Indp_SameFreq_uncond_Eq0}
\begin{split}
    \partial_t \rho & =  \dfrac{\varepsilon}{\tmu+\hmu} \, \nabla_x \cdot\left(\mathbb{V}_q\left(V_{\psi_q}+\dfrac{\tmu}{\hmu}u_{\psi_q}^2\right) \nabla_x\rho\right)
    \end{split}
\end{equation*}
which again has the structure of a diffusion equation, but with an enhanced effective diffusivity due to the decoupling between speed and direction processes. Finally, for the case of independent operators with frequencies of different order, the macroscopic equation under scaling $(a)$ becomes trivial, i.e., no dynamics are retained at zero and first order, while under the scaling $(b)$, equation \eqref{macro3.Eq0} reduces to
\begin{equation*}\label{macro3.Eq0_uncond}
\begin{split}
    &\partial_t \rho =\dfrac{\varepsilon}{\hmu}\nabla_x\cdot\nabla_x \cdot \left[u_{\psi_q}^2\mathbb{V}_q\,\rho\right],
    \end{split}
\end{equation*}
confirming once again that, in the absence of directional bias, all transport effects vanish, and the evolution is governed solely by a diffusion term. In this case, due to the faster speed dynamics, the effective diffusivity is governed by the squared mean speed $u^2_q$ and the directional variance-covariance tensor $\mathbb{V}_q$.



\section{Numerical tests}\label{sec:num_tests}

In this section, we perform some numerical tests aimed at qualitatively comparing the evolution of the standard linear transport equation \eqref{kin.eq:M} with the evolution of the linear transport equation with distinct scattering operators \eqref{eq:kin_strong} introduced in these notes. Precisely, we numerically integrate the transport equations to approximate the density distribution $f$, as in \cite{loy2019JMB,conte2023SIAP,conte2022multi}, and in turn the corresponding macroscopic density via \eqref{def:bf}. We present two numerical tests.
\begin{itemize}
    \item[{\bf Test 1.}] In section \ref{sec:num_test1}, we compare the evolution of linear transport equation in three cases: Eq.~\eqref{kin.eq:M} for $M(\tv,\hv)=\psi(\tv)q(\hv)$, i.e., a single scattering operator with unconditioned speed transition probability; Eq.~\eqref{kin.eq:M} for $M(\tv,\hv)=\psi(\tv|\hv)q(\hv)$, i.e., a single scattering operator with conditioned speed transition probability; Eq.~\eqref{eq:kin_strong} with two distinct scattering operators and conditioned speed transition probability $\psi(\tv|\hv)$.
    \item[{\bf Test 2.}] In section \ref{sec:num_test2}, we compare the evolution of the macroscopic settings derived in section \ref{sec:macro_lim} via two demonstrative examples.
\end{itemize}

\subsection{Test 1: comparison between the linear transport equations}\label{sec:num_test1}
Test 1 is designed to highlight the qualitative differences in the evolution of the described kinetic model in a particular experimental setting. Formally, we are dealing with:
\begin{itemize}
    \item[(K1)] Eq.~\eqref{kin.eq:M} in which $M(\tv,\hv)=\psi(\tv)q(\hv)$, namely there is no dependence between the new speed and the current direction;
    \item[(K2)]  Eq.~\eqref{kin.eq:M} with $M(\tv,\hv)=\psi(\tv|\hv)q(\hv)$, which introduces a conditional dependence of the new speed on the current direction;
    \item[(K3)] Eq.~\eqref{eq:kin_strong} with $\psi=\psi(\tv|\hv)$ in the scattering operator for the speed. 
\end{itemize}
We shall perform the test in a two dimensional setting corresponding to $d=2$, so that we can define ${\hv=(\cos(\theta)\sin(\theta))}$, $\theta \in [0,2\pi)$. In particular, we consider the region $\Omega=[0, 2.5]\times[0,2.5]$. For the transition probability of the speed, we consider the following unimodal von Mises distribution rescaled over $[0, U]$, i.e.,
\begin{equation}\label{psi_unimodal}
    \psi(\tv|\hv)=\dfrac{1}{2\pi I_0(k_\psi)}\exp\left[k_\psi \cos\left(2\pi\dfrac{\tv-\bar{U}(\hv)}{U}\right)\right],
\end{equation}
where the maximum speed is $U=4$, the concentration parameter $k_\psi=80$, while $I_0(k_\psi)$ is the Bessel function of order 0. The mean speed $\bar{U}(\hv)$ is chosen to be constant $\bar{U}(\hv)=\bar{U}=1.5$ in the case (K1), while for (K2) and (K3) we set
\begin{equation*}
    \bar{U}(\hv)=
    \begin{sistem}
       1.5\qquad \text{for}\,\,\,\, 0\le\hv<\pi\,,\\[0.2cm]
       0.2\qquad \text{for}\,\,\,\, \pi\le\hv<2\pi\,.
    \end{sistem}    
\end{equation*}
For the transition probability of the directions, we use a bimodal von Mises distribution, with given concentration parameter $k_q=10$ and preferential direction of migration $\theta_q=\pi/2$, i.e.,
\begin{equation}\label{q_bimodal}
q(\hv)=\dfrac{1}{4\pi I_0(k_q)}\Big(\exp\left[k_q \cos(\theta-\theta_q)\right]+\exp\left[-k_q \cos((\theta-\theta_q)\right]\Big)\,.
\end{equation}
We remark that in this case $u_q=0$.
The initial condition $\rho_0$ is a Gaussian
\begin{equation}\label{rho0}
\rho_0(x)=r_0\exp\left(-\dfrac{(x-x_0)^2}{2\sigma^2}\right)    
\end{equation}
with $r_0=1$, $\sigma^2=0.01$, and $x_0=(1.25,1.25)$. 

This experimental setting is inspired by various realistic scenarios in which the particles exhibit a preferential direction of movement, yet their forward and backward speeds are clearly different. A typical example is traffic on a straight highway: although the preferential direction is well-defined, vehicles can move forward at high speed but cannot realistically move backward, even if the direction (in a symmetric sense) is the same. A second example, drawn from a biological context, involves cell migration in a fibrous environment with steric hindrance effects. In this case, although the probability of selecting direction $\theta_q$ or $-\theta_q$ may be symmetric, the density of the extracellular matrix could impede cell movement in one direction, thereby enhancing it in the opposite one.

Results of this first test are shown in Figure \ref{fig:test1}.
\begin{figure}[h!]
       \centering
       \includegraphics[width=\linewidth]{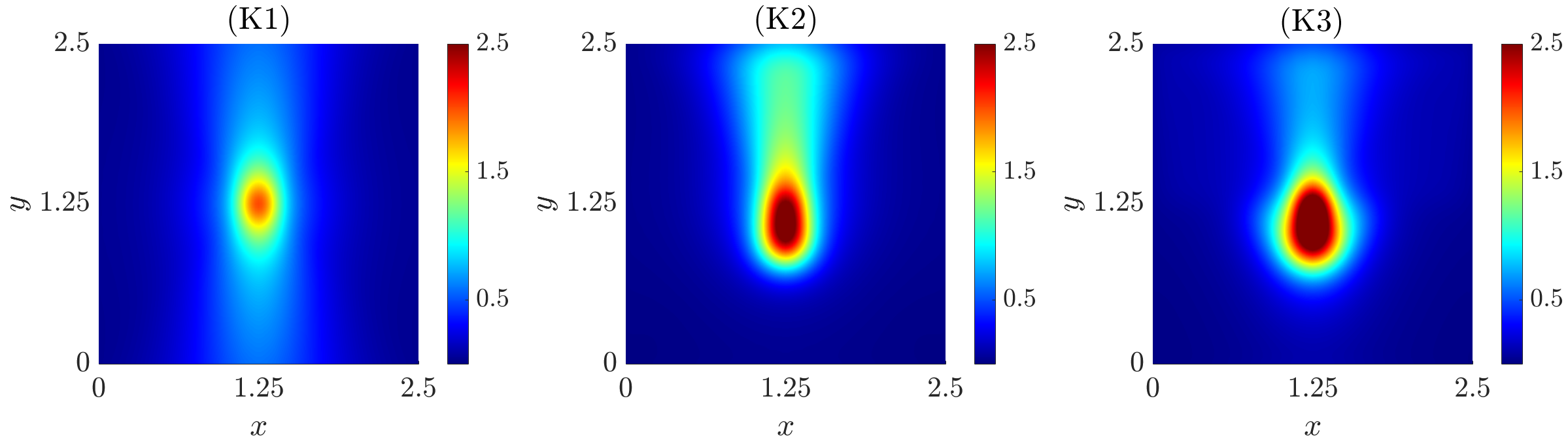}
       \caption{{\bf Test 1.} Evolution of the macroscopic density $\rho$ in the three cases (K1), (K2), and (K3) at time $T=1.88$ and starting from the initial distribution defined in \eqref{rho0}.}
       \label{fig:test1}
\end{figure}
We observe a clear distinction between the unconditioned case (K1) and the conditioned cases (K2) and (K3). Specifically, when the mean speed is constant and independent of the particles' current direction, as in (K1), the dynamics are primarily governed by the fixed directional distribution $q(\hv)$. As a result, particles tend to move away from their initial position predominantly along the directions $\theta_q$ and $-\theta_q$.  In contrast, when we assume that $\psi=\psi(\tv|\hv)$, as in (K2) and (K3), the movement changes significantly:  between the two symmetric directions, particles tend to move more strongly in the directions where the mean speed $\bar{U}$ is higher, namely for directions within the interval $[0,\pi)$. Regarding cases (K2) and (K3), we observe a qualitatively similar trend: in both cases, particles tend to move predominantly along the direction $\theta_q=\pi/2$, where the mean speed is higher. However, the dynamics in (K2) appear to be faster compared to (K3). Additionally, the ellipsoidal shape of the spreading pattern, which characterizes the dynamics driven by $q(\hv)$, is more elongated along the y-axis in (K2). In (K3), while the y-axis component still dominates, the spread exhibits a comparatively larger x-axis component than in (K2). These differences may be attributed to the structural distinction in the models. In fact, in (K3), the use of two distinct scattering operators for speed and direction allows for changes in speed without necessarily altering direction. Given our specific choice of $\psi$ and $\bar{U}(\hv)$, this decoupling permits particles to explore regions of higher speed even if they are not perfectly aligned with the preferential direction prescribed by $q(\hv)$.

\subsection{Test 2: comparison between the macroscopic settings}\label{sec:num_test2}
This test is designed to highlight the qualitative differences in the evolution of the macroscopic models derived in section~\ref{sec:macro_lim}, under different assumptions on the microscopic transition frequencies. In this test, we focus exclusively on the case of conditioned transition probability $\psi(\tv|\hv)$ for the speed, and we consider the following four macroscopic models:
\begin{itemize}
    \item[(M1)] Eq.~\eqref{macro1_govern_dep} derived from the transport equation~\eqref{kin.eq:M} with a single scattering operator, under the scaling $\mu\to\mu/\varepsilon$;
    \item[(M2)]  Eq.~\eqref{macro1_govern_corr} derived from the transport equation~\eqref{eq:kin_strong} with distinct scattering operators for speed and direction and assuming jump processes with the same order frequencies, i.e., $\tmu\to\tmu/\varepsilon$ and $\hmu\to\hmu/\varepsilon$;
    \item[(M3)] Eq.~\eqref{macro2.govern_corr_baru}, also derived from the transport equation~\eqref{eq:kin_strong} with distinct scattering operators for speed and direction, but assuming a faster jump process for the direction, i.e., $\tmu\to\tmu/\varepsilon$ and $\hmu\to\hmu/\varepsilon^2$;
    \item[(M4)] Eq.~\eqref{macro3.govern_corr} derived from the transport equation~\eqref{eq:kin_strong} with distinct scattering operators for speed and direction but under the opposite assumption, i.e., a faster speed jump process with $\tmu\to\tmu/\varepsilon^2$ and $\hmu\to\hmu/\varepsilon$.
\end{itemize}
We consider two experimental configurations, referred to as {\bf Test 2.1} and {\bf Test 2.2}. 

In {\bf Test 2.1}, we adopt the same initial conditions defined in section~\ref{sec:num_test1}. The goal here is to compare the evolution of the four macroscopic models (M1)–(M4) and assess the influence of the underlying microscopic assumptions on the large-scale behavior. The results of this test are shown in Figure \ref{fig:test2.1}.
\begin{figure}[h!]
       \centering
       \includegraphics[width=\linewidth]{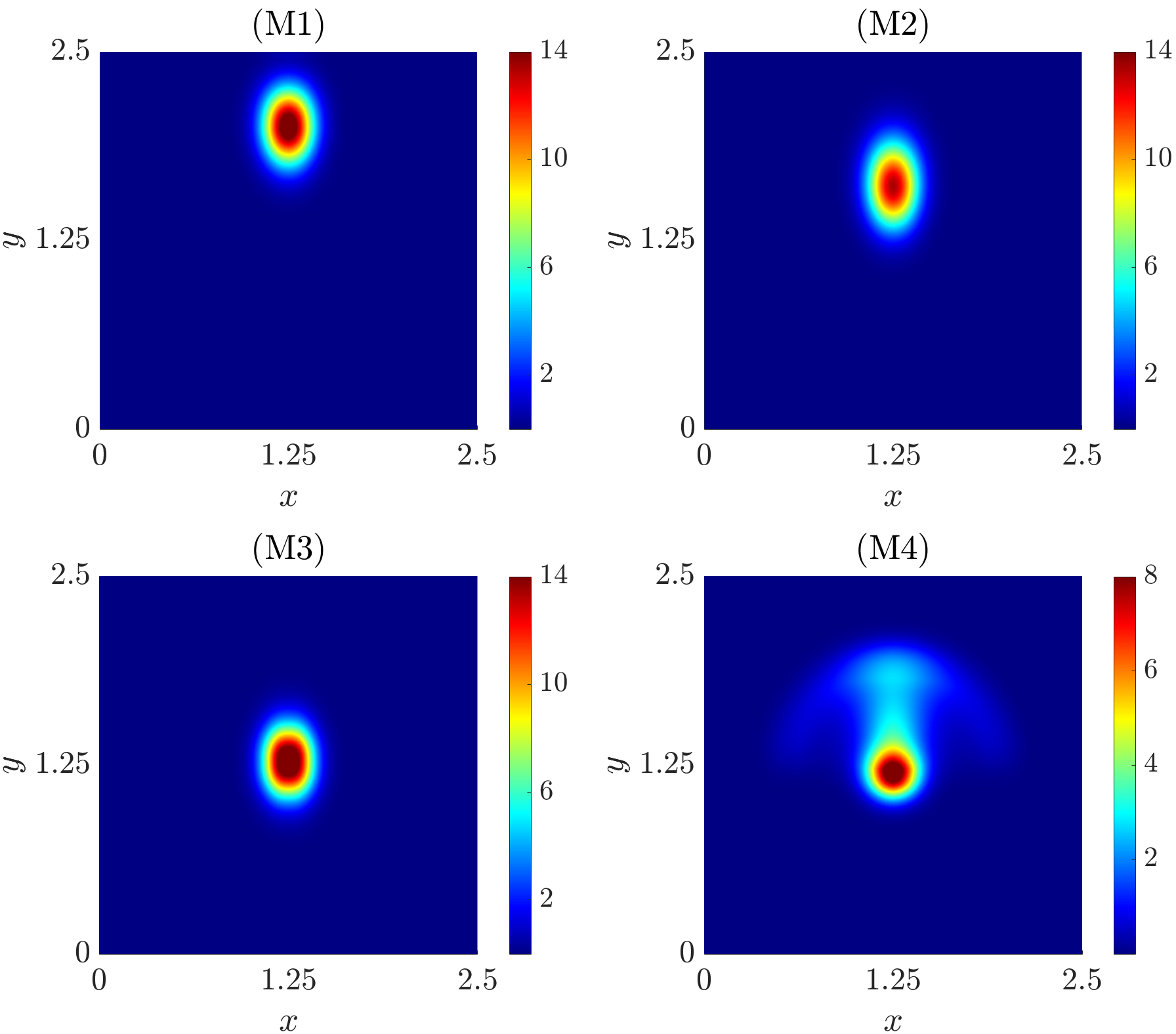}
       \caption{{\bf Test 2.1.} Evolution of the macroscopic density $\rho$ in the settings (M1), (M2), (M3), and (M4). Precisely, (M1) and (M2) are shown at time $T=1.25$, (M3) at time $T=4.69$, and (M4) at time $T=0.47$. For all the simulations, we set $\mu=\tmu=\hmu=1$.}
       \label{fig:test2.1}
\end{figure}
We first observe that, qualitatively, the macroscopic models (M1) and (M2) exhibit similar evolutionary behavior, both showing the expected drift-driven dynamics along the preferential direction $\theta_q=\pi/2$. The differences in their propagation speed can be mainly attributed to the nature of the leading-order drift terms. In (M1), the drift is governed solely by $u_M$, whereas in (M2) it results from a combination of $u_M$ and $u_{\psi_q}u_q$ weighted by a factor of 1/2. Since $u_q=0$, the drift velocity in (M2) is the half with respect to the case (M1). This structural difference leads to a reduced effective drift in (M2) compared to (M1), thereby providing an explanation to the slightly slower dynamics. Looking instead at the evolution of models (M3) and (M4), we observe two distinct behaviors. When the direction jump process relaxes faster, i.e., when $\hmu\to\hmu/\varepsilon^2$ as in (M3), the dynamics are strongly influenced by the equilibrium distribution of the scattering operator for direction. This is evidenced by the ellipsoidal shape of the spreading pattern, which aligns with the preferential direction $\theta_q$. However, the overall dynamics are significantly slower compared to the other models. In contrast, when the speed jump process relaxes faster, i.e., $\tmu\to\tmu/\varepsilon^2$, as in (M4), the distribution of speed plays a more dominant role. In this case, particles tend to explore a broader portion of the domain, particularly regions where the mean speed is higher (i.e., $\hv\in[0,\pi)$). Nevertheless, a greater accumulation of particles is still observed along the preferential direction $\theta_q$, consistent with the influence of the directional bias encoded in $q(\hv)$.

In {\bf Test 2.2}, we consider the same region $\Omega=[0, 2.5]\times[0,2.5]$.
For the transition probability of the speed, we use the unimodal von Mises distribution rescaled over $[0, U]$, as given in \eqref{psi_unimodal}, with a direction-dependent mean speed defined by $\bar{U}(\hv)=1.5|\cos(\hv)|$ and set the concentration parameter to $k_\psi=150$. For the transition probability of the directions, we adopt the same bimodal von Mises distribution as in \eqref{q_bimodal}, but we set the concentration parameter to $k_q=2$. 

This experimental setting was designed with the aim of defining a conditioned speed distribution that, given a fixed directional distribution with preferential  orientation $\theta_q$, counterbalances this bias by encouraging movement in the perpendicular direction. The goal is to make an otherwise anisotropic movement more isotropic. This approach draws biological inspiration from the idea of controlling and regulating the invasive behavior of brain tumor cells. A well-established phenomenon is related to the tendency of cells to migrate preferentially along the white matter tracts of the brain, resulting in highly anisotropic movement. This directional migration contributes to deeper tumor infiltration and makes the invasion more challenging to predict and treat. Certain therapeutic strategies aim to disrupt this guided migration by slowing down cell migration and thus promoting a more isotropic and uniform cell movement, thereby making tumor spread more regular, predictable, and potentially easier to control. The therapeutic strategy (embodied by $\tilde{\mathcal{L}}$) is of course independent and happens on a different time scale than the one which rules cell migration, which is embodied by $\hat{\mathcal{L}}$.

Results of this test are shown in Figure \ref{fig:test2.2}.
\begin{figure}[h!]
       \centering
       \includegraphics[width=\linewidth]{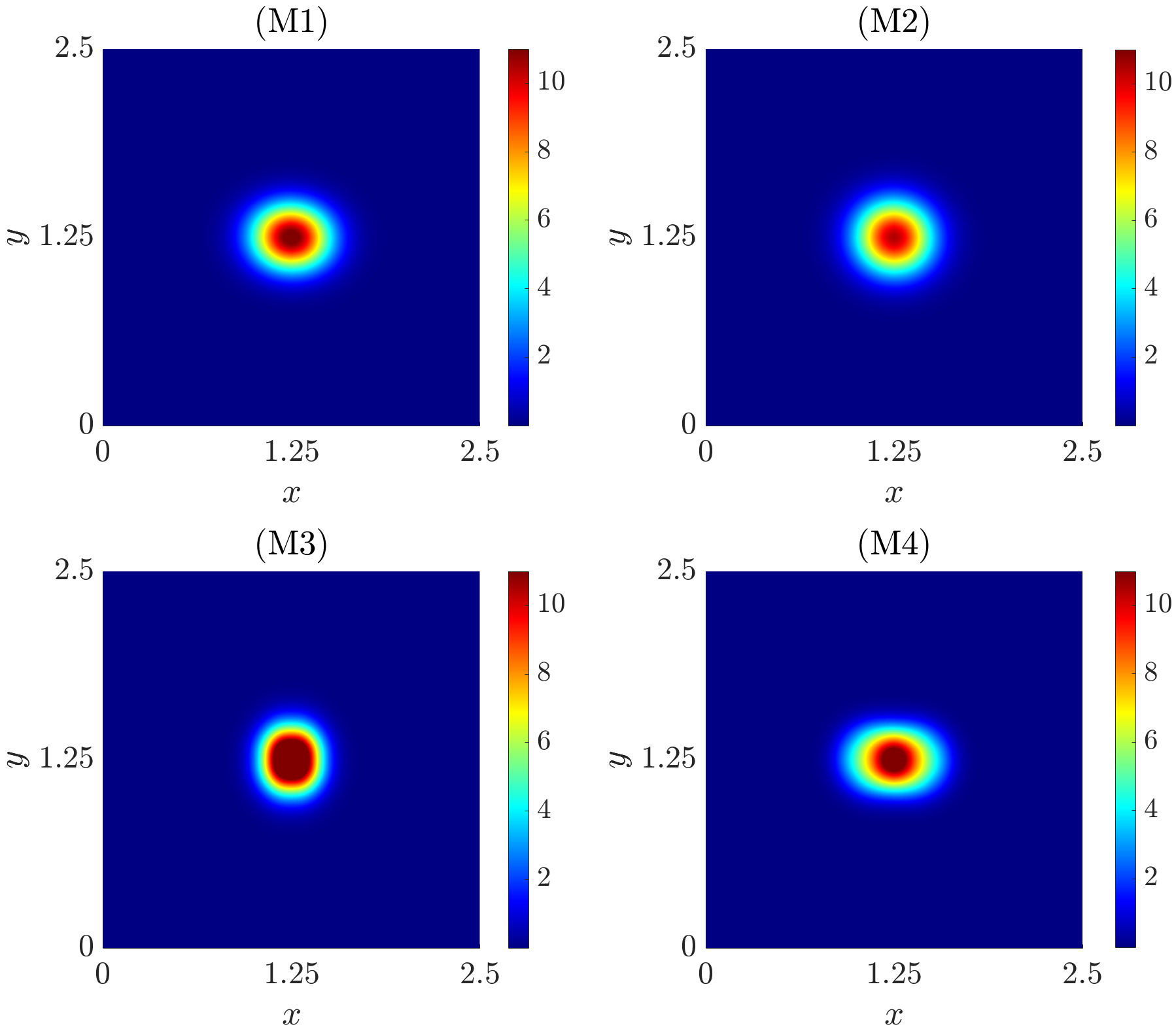}
       \caption{{\bf Test 2.2.} Evolution of the macroscopic density $\rho$ in the settings (M1), (M2), (M3), and (M4). Precisely, (M1) and (M2) are shown at time $T=1.875$, (M3) at time $T=5.94$, and (M4) at time $T=0.1875$. For all the simulations, we set $\mu=\tmu=\hmu=1$.}
       \label{fig:test2.2}
\end{figure}
Starting from the evolution of models (M1) and (M2), we observe that the choice of the transition probability for the speed and its associated concentration parameter enables the recovery of an isotropic-like movement of the particles. This counterbalances the effect of the directional transition probability, which on its own would have resulted in a purely anisotropic motion in direction $\theta_q$. As previously noted, the primary difference between the two simulations concerns the duration of the transient regimes. In particular, model (M2) exhibits faster spreading dynamics compared to model (M1). As before, this may be attributed to differences in the structure of the drift terms and, more specifically, to the greater influence of the speed-jump process, which, in this case, tends to drive particles preferentially along the $x-$direction. In contrast, models (M3) and (M4) exhibit markedly different behaviors. In the case of model (M3), where $\hmu\to\hmu/\varepsilon^2$, the dynamics are strongly influenced by the equilibrium distribution of the scattering operator acting on direction, which induces a more pronounced ellipsoidal shape in the spreading pattern, elongated along the $y-$direction. Although anisotropy remains present, it is less dominant than in the previous case, primarily due to the altered transition probability governing speed. On the other hand, in model (M4), where $\tmu\to\tmu/\varepsilon^2$, the speed distribution plays a central role. This leads to a distinct behavior, characterized by enhanced particle propagation in regions of the domain with higher mean speeds, with a more pronounced anisotropic spreading pattern and an ellipsoidal shape elongated along the $x-$direction, oriented oppositely to that observed in model (M3).

\section{Conclusion}
In this work, we proposed a novel linear transport model describing the evolution of a one particle distribution undergoing free particle drift and incorporating two distinct scattering mechanisms to describe changes in its speed and direction. The model is grounded in the classical theory of the linear Boltzmann equation for particle dynamics. However, the presence of two separate scattering operators, acting at different times and with different intensities, necessitated a non-trivial analysis of their properties, posing non-standard challenges. 

Starting from a microscopic description of the dynamics based on discrete-time stochastic processes, we modeled the evolution of the microscopic state of the particles. Specifically, the dynamics of speed and direction were described by two distinct jump processes occurring at different stochastic times and with different frequencies. From this microscopic framework, we formally derived the novel kinetic equation for the one-particle density function $f$, along with the ones for the associated marginal distributions for speed ($\tf$) and direction ($\hf$). We examined the structure of the operators appearing in the corresponding kinetic equations, highlighting both their connections and distinctive features. Furthermore, we analyzed the resulting kinetic equation using semigroup theory, which provided insights into the behavior of continuous stochastic trajectories. In particular, we highlighted the main differences between the proposed description and the standard case of a single scattering operator for the velocity changes \eqref{kin.eq:M}.

We investigated the properties of the newly introduced scattering operators for speed and direction by analyzing their kernels and stationary states. Specifically, we characterized the structure of the kernel for the operators of both the one-particle distribution $f$ and its marginals, proving that the kernel of $\mathcal{L}$ in the equation for $f$ is included in the intersection of the kernels of the operators $\hat{\mathcal{L}}_q$ and $\tilde{\mathcal{L}}_\psi$ in the equations for $\hf$ and $\tf$. Moreover, we computed the stationary states for the three distributions, emphasizing the differences between the unconditioned case, where $\psi=\psi(\tv)$ and the asymptotic distribution $f^{\infty}$ factorizes into speed and direction components, and the conditioned case, where such factorization no longer holds. We then analyzed the Fredholm-alternative and the non-standard pseudo-inverse operators in suitable Hilbert spaces, establishing the boundedness of the operator $\mathcal{L}$, and proving a novel entropy decay result in the spatially homogeneous setting for the unconditioned case $\psi=\psi(\tv)$. As discussed in section \ref{sec:entropy}, extending the entropy decay analysis to the general case $\psi=\psi(\tv|\hv)$ remains an open question that we aim to address in future research.

Building on these results, we formally derived the macroscopic limits of the novel kinetic equation under different scaling regimes. We considered both the case where the frequencies of the jump processes for speed and direction are of the same order, and the case where one process is faster than the other one. The derived macroscopic equations were analyzed and we provided an interpretation in some simplified scenarios, such as when $\tmu=\hmu=1$ or when $\psi=\psi(\tv)$. We then focused on the case where the directional distribution has zero mean, i.e., $u_q=0$, comparing the different macroscopic equations. This scenario may be particularly relevant for biological applications involving cell migration in fibrous environments. In such settings, fibers are typically assumed to be non-polarized, and their distribution naturally satisfies the condition $u_q=0$. Cells tend to migrate along these fibers, while their speed may be modulated by other factors operating on different time scales. This behavior is effectively captured by the proposed model, making it well-suited to describe cell motility in such structured, anisotropic media.

Finally, the numerical results presented in section \ref{sec:num_tests} provided a visual illustration of the key differences between the standard linear transport equation \eqref{kin.eq:M}, with a single scattering operator for speed and direction, and the newly introduced formulation. Our simulations showed that, with the new model, particles exhibit greater flexibility in exploring regions of the domain where either speed or direction is favorable, rather than being predominantly guided by a single combined influence of both. This highlights the richer dynamics captured by decoupling the scattering processes for speed and direction.

We remark that, although the model and the analysis were developed in a general framework of particle movement influenced by an external medium, the approach is applicable to a wide range of scenarios. One notable application is in the study of traffic dynamics, where vehicles (modeled as particles) can independently adjust their direction and speed for different reasons. For example, a change in direction may result from the driver's intention to reach a specific destination, while speed variations may be caused by traffic conditions or imposed speed limits. Furthermore, as previously discussed, the model is well-suited to describe various experimental settings involving cell migration in heterogeneous environments, where distinct chemical or mechanical factors may influence directional changes and speed modulation on different time scales.

Although the results presented here offer a rich analysis of the characteristic features of the newly introduced operators and their differences with respect to the standard case, several open questions remain, paving the way for future research. Future investigations include the well-posedness, existence and uniqueness of the model, macroscopic limits in other, possibly fractional, regimes, but also hypocoercivity properties of the model in presence of spatially heterogeneous scattering kernels. 




\appendix
\section{Appendix} 
\subsection{Formal derivation of the linear transport equation}\label{appendix1}
Le us consider a test function $\phi$, which is an observable quantity depending on the microscopic variables, i.e., $\phi=\phi(x,\tv,\hv)$ defined on $\R^d\times \mathcal{V}$, being $\mathcal{V}= \mathbb{S}^{d-1}\times \R_+$. We assume that ${\phi \in \mathcal{C}_c^{\infty}(\R^d\times \mathcal{V})}$ with compact support. Our purpose is to derive the evolution equation for the expected value of $\phi$ along the stochastic trajectories defined by~\eqref{def:micro.X}-\eqref{def:micro.V}-\eqref{def:micro.hV}-\eqref{def:Berny}-\eqref{ass:indep1}-\eqref{ass:indep2}-\eqref{ass:indep3}-\eqref{def:psi}-\eqref{def:q}. To this aim, we calculate the mean variation rate of the quantity $\phi$ in the time interval $\Delta{t}$, namely
\begin{equation*}\label{def:incr.ratio}
\frac{\ave{\phi\left(X_{t+\Delta t},\tV_{t+\Delta t}, \hV_{t+\Delta t}\right)}-\ave{\phi\left(X_t,\tV_t,\hV_t\right)}}{\Delta{t}}\,,
\end{equation*}
where, here and henceforth, $\ave{\cdot}$ denotes the expectation with respect to all random variables appearing in the brackets. Specifically, if $\eta\in\mathcal{B}\subseteq\R$ is a random variable with law $h=h(\eta):\mathcal{B}\to\R_+$ then

$$ \ave{\cdot}:=\int_\mathcal{B}(\cdot)h(\eta)\,d\eta. $$
Considering~\eqref{def:micro.X}-\eqref{def:micro.V}-\eqref{def:micro.hV}, we have that
\begin{align*}
	&\frac{\ave{\phi\left(X_{t+\Delta t},\tV_{t+\Delta t}, \hV_{t+\Delta t}\right)}}{\Delta{t}}= \\[0.2cm]
	&\frac{\ave{\phi\left(X_{t}+ \tV_t\hV_t\Delta t,(1-\tilde{\Theta}_{t+\Delta t})\tV_{t} + \tilde{\Theta}_{t+\Delta t} \tV_{t+\Delta t}',(1-\hat{\Theta}_{t+\Delta t})\hV_{t} + \hat{\Theta}_{t+\Delta t} \hV_{t+\Delta t}'\right)}}{\Delta{t}}=\\[0.2cm]
&(1- \tmu \Delta t) (1-\hmu\Delta t)\frac{\ave{\phi\left(X_{t}+ \tV_t\hV_t\Delta t,\tV_{t} ,\hV_{t} \right)}}{\Delta{t}}+ \tmu  \hmu\Delta t^2 \frac{\ave{\phi\left(X_{t}+ \tV_t\hV_t\Delta t,\tV_{t}' ,\hV_{t}' \right)}}{\Delta{t}} \\[0.2cm]
&+ \tmu \Delta t (1-\hmu\Delta t)\frac{\ave{\phi\left(X_{t}+ \tV_t\hV_t\Delta t,\tV_{t}' ,\hV_{t} \right)}}{\Delta{t}}+(1- \tmu \Delta t) \hmu\Delta t \frac{\ave{\phi\left(X_{t}+ \tV_t\hV_t\Delta t,\tV_{t} ,\hV_{t}' \right)}}{\Delta{t}},
	\end{align*}
	where the second equality follows from~\eqref{def:Berny} and from the assumed independence~\eqref{ass:indep1}.
As a consequence, we have that
\begin{equation}\label{eq:pass1}
\begin{aligned}[b]
	&\frac{\ave{\phi\left(X_{t+\Delta t},\tV_{t+\Delta t}, \hV_{t+\Delta t}\right)}-\ave{\phi\left(X_t, \tV_t,\hV_t\right)}}{\Delta{t}}= \frac{\ave{\phi\left(X_{t}+ \tV_t\hV_t\Delta t,\tV_{t} ,\hV_{t} \right)}-\ave{\phi\left(X_t, \tV_t,\hV_t\right)}}{\Delta{t}} \\[0.4cm]
&+ \tmu \ave{\phi\left(X_{t}+ \tV_t\hV_t\Delta t,\tV_{t}' ,\hV_{t} \right)}+ \hmu\ave{\phi\left(X_{t}+ \tV_t\hV_t\Delta t,\tV_{t} ,\hV_{t}' \right)}-( \tmu+\hmu)\ave{\phi\left(X_{t}+ \tV_t\hV_t\Delta t,\tV_{t} ,\hV_{t} \right)}\\[0.4cm]
&+ \tmu  \hmu\Delta t \Big(\ave{\phi\left(X_{t}+ \tV_t\hV_t\Delta t,\tV_{t}' ,\hV_{t}' \right)} -\ave{\phi\left(X_{t}+ \tV_t\hV_t\Delta t,\tV_{t}' ,\hV_{t} \right)}-\ave{\phi\left(X_{t}+ \tV_t\hV_t\Delta t,\tV_{t} ,\hV_{t}' \right)}\Big)\,.
\end{aligned}
\end{equation}
Considering the limit of $\Delta{t}\to 0^+$, the left-hand side gives the instantaneous time variation of the expected value of $\phi$, i.e.,
\begin{equation*}
\frac{\ave{\phi\left(X_{t+\Delta t},\tV_{t+\Delta t}, \hV_{t+\Delta t}\right)}-\ave{\phi\left(X_t,\tV_t,\hV_t\right)}}{\Delta{t}}\underset{\Delta{t}\to 0^+}{\longrightarrow}\frac{d}{dt}\ave{\phi\left(X_t, \tV_t,\hV_t\right)}\,.
\end{equation*}
The first term on the right-hand side in~\eqref{eq:pass1} gives
\[
\frac{\ave{\phi\left(X_{t}+ \tV_t\hV_t\Delta t,\tV_{t},\hV_t\right)-\phi\left(X_t, \tV_t,\hV_t\right)}}{\Delta{t}}\underset{\Delta{t}\to 0^+}{\longrightarrow}\ave{\nabla_{x}\cdot \left( \tV_t\hV_t\phi(X_t, \tV_t,\hV_t)\right)}\,,
\]
while the second line in~\eqref{eq:pass1} tends to
\[
 \tmu \ave{\phi\left(X_{t},\tV_{t}' ,\hV_{t} \right)}+ \hmu\ave{\phi\left(X_{t},\tV_{t} ,\hV_{t}' \right)}-( \tmu+\hmu)\ave{\phi\left(X_{t},\tV_{t} ,\hV_{t} \right)}\,.
\]
In the third line in~\eqref{eq:pass1}, the first term in brackets is a gain term regarding simultaneous reorientations and speed-jumps, while the second and third terms are loss terms in which either the direction or the speed changes. These ones are higher order terms in $\Delta t$ and vanish as $\Delta t \to 0^+$.
We now write explicitly the expected values of interest 
\[
\ave{\phi\left(X_{t},\tV_{t} ,\hV_{t} \right)} = \int_{\R^d}\int_\mathcal{V} \phi(x,\tv,\hv) f(t,x,\tv,\hv) \, {\rm d}\tv \, {\rm d}\hv \, {\rm d}x\,,
\]

\[
\ave{\phi\left(X_{t},\tV_{t}' ,\hV_{t} \right)} = \int_{\R^d}\int_\mathcal{V} \intR\phi(x,\tv',\hv) \psi(\tv'|\hv) \, {\rm d}v'f(t,x,\tv,\hv) \, {\rm d}\tv \, {\rm d}\hv \, {\rm d}x\,,
\]

\[
\ave{\phi\left(X_{t},\tV_{t} ,\hV_{t}' \right)} = \int_{\R^d}\int_\mathcal{V} \intS\phi(x,\tv,\hv') q(\hv') {\rm d}\hv' f(t,x,\tv,\hv) \, {\rm d}\tv \, {\rm d}\hv \, {\rm d}x\,,
\]
and we obtain
\begin{equation*}
\begin{aligned}[b]
\langle\nabla_{x}\cdot \left( \tV_t\hV_t\phi(X_t, \tV_t,\hV_t)\right)\rangle&=\int_{\R^d}\intV \nabla_x \cdot \left( \tv\hv \phi(x,\tv,\hv)\right)\, f(t,x,\tv,\hv) \, {\rm d}\hv {\rm d}\tv  dx\\[0.2cm]
&=\int_{\R^d}\intV\nabla_x \cdot \left( \tv\hv \phi(x,\tv,\hv)\, f(t,x,\tv,\hv)\right) \, {\rm d}\hv {\rm d}\tv dx\\[0.2cm]
&\phantom{=}-\int_{\R^d} \intV\phi(x,\tv,\hv) \tv \hv \cdot \nabla_x  f(t,x,\tv,\hv) \, {\rm d}\hv {\rm d}\tv dx\\[0.2cm]
&=-\int_{\R^d}\intV\phi(x,\tv,\hv) \tv \hv \cdot \nabla_x  f(t,x,\tv,\hv) \, {\rm d}\hv {\rm d}\tv dx
\end{aligned}
\end{equation*}
where the disappearance of the term in the last passage is due to the Divergence Theorem.


Next, choosing the test function such that it factorizes the dependency on $x$ and the couple $(\tv,\hv)$, i.e. $\phi(x,\tv,\hv)=\xi(x) \varphi(\tv,\hv)$, it is possible to eliminate the integration in $x$ and of the $x$-test function $\xi(x)$ and write the weak form of the equation expressing the evolution of the expectation of $\varphi(\tv,\hv)$ as follow:
\begin{equation*}
\begin{aligned}[b]
&\dfrac{d}{dt}\intV\varphi(\tv,\hv) f(t,x,\tv,\hv) {\rm d}\hv {\rm d}\tv+\nabla_x \cdot \intV\varphi(\tv,\hv) \hv\tv f(t,x,\tv,\hv) \, {\rm d}\hv \, {\rm d}\tv=\\[0.2cm]
& \tmu\intV\left[\intR\psi(\tv'|\hv)\varphi(\tv',\hv)\, {\rm d}\tv'-\varphi(\tv,\hv)\right]f(t,x,\tv,\hv) \,{\rm d}\hv \, {\rm d}\tv\\
&+\hmu\intV\left[\intS q(\hv')\varphi(\tv,\hv')\, {\rm d}\hv'-\varphi(\tv,\hv)\right]f(t,x,\tv,\hv) \,{\rm d}\hv \, {\rm d}\tv,
\end{aligned}
\end{equation*}
which is ~\eqref{eq:kin_weak}.

\subsection{Additional details of the derivation of the macroscopic models}\label{appendix2}

Taking into account the derivation of the macroscopic limit and, particularly, the case of jump processes where the frequencies of speed and direction changes are of the same order, i.e., $\tmu \to \tmu/\varepsilon$ and $\hmu \to \hmu\varepsilon$, we report here additional calculations related to the correction term $f^\perp$ in the Chapman-Enskog expansion of $f^\varepsilon$. In particular, after deriving the expression for $f^\perp$ as given in~\eqref{f_perp.2}, we proceed to compute the following integral term:
\begin{allowdisplaybreaks}
    \begin{align*}
        &\intV\nabla_x\cdot\left(v f^{\bot}\right){\rm d}v=\\[0.3cm]
        &-\dfrac{1}{\tmu+\hmu}\nabla_x\cdot\nabla_x\cdot\left[\mathbb{D}_T\rho\right]+\dfrac{1}{\hmu}\nabla_x\cdot\left[u_T\nabla_x\cdot(\rho u_T)\right]+\dfrac{\hmu}{\tmu(\tmu+\hmu)}\nabla_x\cdot\left[u_{\psi_q}u_q\nabla_x\cdot(\rho u_T)\right]\\[0.3cm]
        &-\dfrac{\tmu^2}{\hmu(\hmu+\tmu)^2}\nabla_x\cdot\left[\intS u_\psi^2(\hv)q(\hv)\hv\otimes\hv{\rm d}\hv\,\nabla_x\rho\right]\\[0.3cm]
        &-\dfrac{\tmu}{(\hmu+\tmu)^2}\nabla_x\cdot\left[u_{\psi_q}\intS u_\psi(\hv)q(\hv)\hv\otimes\hv{\rm d}\hv\,\nabla_x\rho\right]\\[0.3cm]
        &-\dfrac{\hmu}{(\hmu+\tmu)^2}\nabla_x\cdot\left[u_q\otimes\intS V_\psi(\hv)\hv q(\hv) {\rm d}\hv\,\nabla_x\rho\right]-\dfrac{\hmu^2}{\tmu(\hmu+\tmu)^2}\nabla_x\cdot[V_{\psi_q} u_q\otimes u_q\nabla_x\rho]\\[0.3cm]
        &-\dfrac{\tmu}{(\hmu+\tmu)^2}\nabla_x\cdot\left[u_q\otimes\intS u_\psi^2(\hv)\hv q(\hv){\rm d}\hv\,\nabla_x\rho\right]\\[0.3cm]
        &-\dfrac{\hmu}{(\hmu+\tmu)^2}\nabla_x\cdot\left[u_{\psi_q}u_q\otimes\intS u_\psi(\hv)\hv q(\hv){\rm d}\hv\,\nabla_x\rho\right]
    \end{align*}
\end{allowdisplaybreaks}
where $\mathbb{D}_T$, $V_\psi(\hv)$, and $V_{\psi_q}$ are defined in \eqref{D_T}, \eqref{V_psi}, and \eqref{V_psiq}, respectively. Then, using the identity 
\begin{equation*}\label{baru_q}
\intS u_\psi(\hv)\hv q(\hv){\rm d}\hv=\dfrac{\tmu+\hmu}{\tmu}u_T-\dfrac{\hmu}{\tmu}u_qu_{\psi_q}\,
\end{equation*}
derived from \eqref{u_T}, and observing that the integral terms is exactly $u_M$ defined in \eqref{u_M}, the final correction term simplifies to: 
\begin{allowdisplaybreaks}
\begin{equation*}\label{f_corr2}
    \begin{split}
        &\intV\nabla_x\cdot\left(v f^{\bot}\right){\rm d}v=\\[0.3cm]
        &-\dfrac{1}{\tmu+\hmu}\nabla_x\cdot\left[V_{\psi_q}\mathbb{V}_q\nabla_x\rho\right]+\dfrac{2\hmu+\tmu}{\hmu(\tmu+\hmu)}\nabla_x\cdot\left[u_T\otimes u_T\nabla_x\rho\right]\\[0.3cm]
        &-\dfrac{\hmu^2}{\tmu(\hmu+\tmu)^2}\nabla_x\cdot\left[(V_{\psi_q}-u_{\psi_q}^2) u_q\otimes u_q\nabla_x\rho\right]\\[0.3cm]
        &-\dfrac{1}{(\hmu+\tmu)^2}\nabla_x\cdot\left[u_q\otimes\intS (\tmu u_\psi^2(\hv)+\hmu V_{\psi_q})\,\hv\, q(\hv){\rm d}\hv\,\nabla_x\rho\right]\\[0.3cm]
        &-\dfrac{\tmu}{\hmu(\hmu+\tmu)^2}\nabla_x\cdot\left[\intS \left(\hmu u_{\psi_q} u_\psi(\hv)+\tmu u_\psi^2(\hv)\right) q(\hv)\hv\otimes\hv\,{\rm d}\hv\,\nabla_x\rho\right]\,.
    \end{split}
\end{equation*}
\end{allowdisplaybreaks}
Moreover, under the assumption that the transition probability for the speed does not depend on the current direction, i.e., $\psi=\psi(\tv)$, the correction term simplifies to
\begin{allowdisplaybreaks}
\begin{equation*}\label{f_corr_uncond}
    \begin{split}\
        &\intV\nabla_x\cdot\left(v f^{\bot}\right){\rm d}v=\\[0.3cm]
        &-\dfrac{1}{\tmu+\hmu}\nabla_x\cdot\left[(V_{\psi_q}\mathbb{V}_q-u_{\psi_q}^2u_q\otimes u_q)\nabla_x\rho\right]-\dfrac{\hmu}{\tmu(\hmu+\tmu)}\nabla_x\cdot\left[V_{\psi_q} u_q\otimes u_q\nabla_x\rho\right]\\[0.3cm]
        &-\dfrac{\tmu}{\hmu(\hmu+\tmu)}\nabla_x\cdot\left[u_{\psi_q}^2\mathbb{V}_q\nabla_x\rho\right]+\dfrac{\tmu^2+\tmu^2}{\hmu\tmu(\hmu+\tmu)}\nabla_x\cdot\left[u_{\psi_q}^2u_q\otimes u_q\nabla_x\rho\right]\,,
    \end{split}
\end{equation*}
\end{allowdisplaybreaks}
which is, then, used into~\eqref{macro1_corr_simpl}.

\subsection*{Acknowledgments}
MC acknowledges support from the PNRR project Young Researchers 2024—SOE {\it ‘Integrated Mathematical Approach to Tumor Interface Dynamics’} (CUP: E13C24002380006), funded by the European Union and the Italian Ministry of University and Research. NL gratefully acknowledges support from the Italian Ministry of University and Research through the Grant PRIN2022-PNRR Project (No. P2022Z7ZAJ) {\it ‘A Unitary Mathematical Framework for Modelling Muscular Dystrophies’} (CUP: E53D23018070001). Both authors acknowledge support by the National Group of Mathematical Physics (GNFM-INdAM).

\subsection*{Data availability statement}

The data that support the findings of this study are available upon reasonable request from the authors.

\subsection*{Competing interests}
The authors have no competing interests to declare that are relevant to the content of this article.

\bibliographystyle{plain}
\bibliography{bibliography}

\end{document}